\documentclass{article}
\usepackage{amsfonts,amsmath,amsthm,amssymb,authblk,latexsym,bm}
\usepackage{hyperref} \hypersetup{colorlinks=true}
\usepackage{citeref}
\usepackage{rotating}

\newtheorem{theorem}{Theorem}[section]
\newtheorem{corollary}[theorem]{Corollary}

\newtheorem{lemma}[theorem]{Lemma}
\newtheorem{definition}[theorem]{Definition}

\newtheorem{remark}[theorem]{Remark}

\numberwithin{equation}{section}

\def\vph{\varphi} \def\vep{\varepsilon}
\def\cda{{F\kern-1pt{*}G}\kern1pt}

\begin{document}
\title{Consta-dihedral Codes over Finite Fields}
\author{
Yun Fan,~ Yue Leng\par
{\small School of Mathematics and Statistics}\par\vskip-1mm
{\small Central China Normal University, Wuhan 430079, China}
}
\date{}
\maketitle

\insert\footins{\footnotesize{\it Email address}:
yfan@mail.ccnu.edu.cn (Yun Fan);
}

\begin{abstract}
It is proved in a reference 
(Fan, Lin, IEEE TIT, vol.67, pp.5016-5025)
that the self-dual (LCD respectively) dihedral codes over
a finite field~$F$ with ${|F|=q}$
are asymptotically good if $q$ is even (odd respectively).
In this paper, we investigate the algebraic property and the asymptotic property
of conta-dihedral codes over $F$, and show that:
if $q$ is even or $4\,|\,(q-1)$,
then the self-dual consta-dihedral codes are asymptotically good;
otherwise,
the LCD consta-dihedral codes are asymptotically good.
And, with the help of a technique developed in this paper,
some errors in the reference mentioned above 
are corrected.

\medskip
{\bf Key words}: Finite fields; dihedral codes; consta-dihedral codes;
self-dual codes; LCD codes.

\end{abstract}

\section{Introduction}

Let $F$ be a finite field with cardinality $|F|=q$,
where $q$ is a power of a prime.
Any $a=(a_{1}, \cdots, a_{n})\in{F}^n$, $a_i\in{F}$, is called a word.
The Hamming weight ${\rm w}(a)$ is defined to be the number
of such indexes $i$ that $a_i\ne 0$.
The Hamming distance between two words $a, a'\in F^n$
is defined as ${\rm d}(a,a')={\rm w}(a-a')$.
Any $\emptyset\ne C\subseteq F^n$ is called a code of length $n$ over $F$;
the words in the code are called code words.
The {\em minimal Hamming distance} ${\rm d}(C)$
is the minimum distance between distinct codewords of~$C$.
If $C$ is a linear subspace of~$F^n$, then $C$ is called a linear code,
 the {\em minimal Hamming weight} ${\rm w}(C)$
is defined to be the minimal weight of the nonzero code words of $C$,
and it is known that ${\rm w}(C)={\rm d}(C)$.
The fraction $\Delta(C)=\frac{{\rm d}(C)}{n}=\frac{{\rm w}(C)}{n}$
is called the {\em relative minimum distance} of~$C$,
and ${\rm R}(C)=\frac{\dim_{F}C}{n}$ is called the {\em rate} of $C$.
A code sequence $C_{1}, C_{2}, \cdots$ is said to be {\em asymptotically good}
if the length $n_{i}$ of $C_{i}$ goes to infinity and both
${\rm R}(C_{i})$ and $\Delta(C_{i})$ are positively bounded from below.
A class of codes is said to be {\em asymptotically good} if
there is an asymptotically good sequence of codes within the class.
The inner product of words $a=(a_1,\cdots,a_{n})$ and
$b=(b_1,\cdots,b_{n})$ is defined to be
$\langle a,b\rangle=\sum_{i=1}^{n}a_ib_i$.
Then the self-orthogonal codes, self-dual codes, LCD codes etc.
are defined as usual, e.g., cf. \cite{HP}.

Let $G$ be a finite group of order $n$. The group algebra
$FG$ 
is the $F$-vector space with basis $G$ and equipped with the multiplication
induced by the multiplication of the group $G$.
Any element $\sum_{x\in{G}}a_{x}x\in FG$ is identified with a word
$(a_x)_{x\in G}\in F^n$.
Then any left ideal of $FG$ is called an {\em $FG$-code}.
If~$G$ is a cyclic (abelian, dihedral, resp.) group,
the $FG$-codes are called {\em cyclic} ({\em abelian, dihedral}, resp.)~{\em codes}.
Any $FG$-submodule of $FG\times FG$ is called a
{\em quasi-$FG$ code of index~$2$}.
If $G$ is cyclic (abelian, resp.),
the quasi-$FG$ codes of index $2$ are also called
{\em quasi-cyclic} ({\em quasi-abelian}, resp.) {\em codes of index~$2$}.

If $G\!=\!\langle x\,|\,x^n\!=\!1\rangle$ is a cyclic group,
then $FG$ is an $F$-algebra generated by~$x$ with a relation $x^n\!=\!1$.
For $0\!\ne\!\lambda\in F$, the $F$-algebra generated by~$x$
with the relation $x^n\!=\!\lambda$ is called a {\em constacyclic group algebra},
and its ideals are called {\em constacyclic codes}.
Next, let $G=\langle u,v\,|\,u^n\!=\!1,v^2\!=\!1, vuv^{-1}\!=\!u^{-1}\rangle$ be
a dihedral group. Then $FG$ is the $F$-algebra (non-commutative)
generated by $u,v$ with three relations $u^n\!=\!1$, $v^2\!=\!1$ and $vu\!=\!u^{-1}v$.
If we replace the relation ``$v^2\!=\!1$'' by the relation ``$v^2\!=\!-1$''
and keep the other two relations invariant,
then the obtained $F$-algebra 
is called a {\em consta-dihedral group algebra}, and its left ideals 
are called {\em consta-dihedral codes}
(cf. Section~\ref{Consta-dihedral group algebras} for details).

It is a long-standing open question (cf. \cite {MW06}):
are the cyclic codes over a finite field asymptotically good?
However, it is well-known long ago that, if the characteristic ${\rm char}(F)=2$,
quasi-cyclic codes of index $2$ over $F$
are asymptotically good, see \cite{CPW,C,K}.
Finite dihedral groups are near to finite cyclic groups,
because a dihedral group of order $2n$ has a normal cyclic subgroup of order $n$.
Bazzi and Mitter \cite{BM} proved that the binary dihedral codes are asymptotic good.
Afterwords, Mart\'inez-P\'erez and Willems \cite{MW} proved the
asymptotic goodness of binary self-dual quasi-cyclic codes of index $2$.
Borello and Willems \cite{BW} proved the
asymptotic goodness of such $FG$-codes that $|F|=p$ is an odd prime
and $G$ is a semidirect product of the cyclic group of order $p$
by a finite cyclic group.

For any finite field $F$, i.e., for any prime power $q$,
in the dissertation \cite{L PhD}
it has been shown that the quasi-cyclic codes of index $2$ over $F$
are asymptotically good; and, if $q$ is even or $4\,|\,(q-1)$
(i.e., $q\,{\not\equiv}\,3~({\rm mod}\;4)$),
the self-dual quasi-cyclic codes of index~$2$ over $F$ are asymptotically good.
Note that self-dual quasi-cyclic codes over $F$ of index~$2$ exist if and only if
$q\,{\not\equiv}\,3~({\rm mod}\;4)$,
cf. \cite{LS03} or \cite[Corollary~IV.5]{LF22}.
Based on Artin's primitive root conjecture, with the same assumption on $q$,
Alahmadi, \"Ozdemir and Sol\'e \cite{AOS} also proved
the asymptotic goodness of the self-dual quasi-cyclic codes of index~$2$.
Lin and Fan \cite {LF22} exhibited further that,
if $q\,{\not\equiv}\,3~({\rm mod}\;4)$,
the self-dual quasi-abelian codes (including the quasi-cyclic case) of index $2$
are asymptotically good.
Recently, Fan and Liu \cite{FL22} discussed
the {\em quasi-constacyclic codes of index $2$} and
showed that such codes are asymptotically good.

On the other hand, Fan and Lin \cite{FL20}
extend the asymptotic goodness of dihedral codes to any q-ary case;
more precisely, they proved that the self-dual dihedral codes (if $q$ is even)
and the LCD dihedral codes (if $q$ is odd)
are asymptotically good. As consequences, the asymptotic goodness
of the self-dual (if $q$ is even) and the LCD (if $q$ is odd)
quasi-cyclic codes of index $2$ are also obtained.
By the way, we observed 
some errors in the proofs of two theorems of the reference \cite{FL20},
as a result of the errors,
\cite[Theorem~IV.5(1)]{FL20} is false;
see Section~\ref{Remarks on} below for details.
Fortunately the issue does not affect the correctness
of the results stated above.

About the relationship between the quasi-cyclic codes of index $2$ and the dihedral codes,
any dihedral (consta-dihedral) code is a quasi-cyclic code of index $2$ in a natural way.
Alahmadi, \"Ozdemir and Sol\'e \cite{AOS} showed that,
if $q$ is even, the self-dual double circulant codes
(a family of self-dual quasi-cyclic codes of index~$2$) are dihedral codes.
Fan and Zhang \cite{FZ22} extended it and showed
a necessary and sufficient condition for a self-dual quasi-cyclic code of index~$2$
being a dihedral code (if $q$ is even)
or a consta-dihedral code (if $q$ is odd).

The research outlines lead us to focus on the consta-dihedral codes.
In this paper we investigate the algebraic property and the asymptotic property
of the consta-dihedral codes, and address the issue in the reference~\cite{FL20}.
To study the algebraic property of consta-dihedral codes,
we develop a technique to evaluate
the orthogonality of (consta-)dihedral codes by matrix computation;
and construct a class of consta-dihedral codes which possess
good algebraic properties (self-orthogonal, LCD etc.).
With the help of this technique we
address the issue in~\cite{FL20} mentioned above
and recover the correct version of the false theorem
of~\cite{FL20} (Theorem~\ref{thm correct} below).
To study the asymptotic property of consta-dihedral codes,
in the class of consta-dihedral codes we constructed,
we count the number of such codes which have bad asymptotic property;
the number is much less than the total quantity of the class,
so that we obtain the following results.

\begin{theorem}[{Theorem~\ref{thm self-dual good} below}]
Assume that $q$ is even or $4\,|\,(q-1)$.
Then the self-dual consta-dihedral codes over $F$ are asymptotically good.
In particular, self-dual quasi-cyclic codes of index $2$ over $F$
are asymptotically good.
\end{theorem}

\begin{theorem}[{Theorem~\ref{thm LCD good} below}]
Assume that $q$ is odd and ${4\!\nmid\!(q-1)}$.
Then the LCD consta-dihedral codes over $F$ are asymptotically good.
In particular, LCD quasi-cyclic codes of index $2$ over $F$
are asymptotically good.
\end{theorem}

In Section~\ref {preliminaries}, some preliminaries are sketched.

In Section~\ref{Consta-dihedral group algebras},
we describe the consta-dihedral group algebras in three ways,
exhibit properties of them; and characterize the structures
of the minimal ideals of consta-dihedral group algebras.

In Section~\ref{Consta-dihedral codes},
with the structures of the consta-dihedral group algebras
characterized in the last section,
we construct a class of consta-dihedral codes,
and characterize their algebraic property precisely
(self-orthogonal, or LCD, etc.).

Section~\ref{Asymptotic property of consta-dihedral codes}
 is devoted to the study on the asymptotic property of consta-dihedral codes.
The two theorems listed above are proved in this section.

In Section~\ref{Remarks on}, we analyze the cause of the issue
in \cite{FL20} mentioned above, with the technique developed in this paper
we address the issue and recover the correct version of
\cite[Theorem IV.5(1)]{FL20}.

Finally, conclusion is made in Section~\ref{Conclusion}.

\section{Preliminaries}\label{preliminaries}

In this paper $F$ is always a finite field with $|F|=q$ which
is a power of a prime, where $|S|$ denotes the cardinality of
any set $S$. And $n>1$ is an integer.

Let $G$ be a finite group, the group algebra
$FG=\{\sum_{x\in{G}}a_{x}x\,|\,a_{x}\in{F}\}$ is the
$F$-vector space with basis $G$ and equipped with
the multiplication induced by the multiplication of the group.
So $FG$ is an $F$-algebra, called the group algebra of $G$ over $F$.
Any $\sum_{x\in{G}}a_{x}x\in{FG}$
is viewed as a word $(a_{x})_{x\in{G}}$ over $F$
with coordinates indexed by $G$.
Any left ideal $C$ of $FG$ is called a group code of $G$ over $F$.
We also say that $C$ is an $FG$-code for short.

It is an anti-automorphism of the group: $G\to G$, $g\mapsto g^{-1}$,
which induces an anti-automorphism of the group algebra:
\begin{align} \label{eq bar map}
  FG\;\longrightarrow\; FG,~~~
  \sum_{g\in G}a_{g} g \; \longmapsto \; \sum_{g\in G}a_g g^{-1}.
\end{align}
We denote $\sum_{g\in G}a_g g^{-1}=\overline{\sum_{g\in G}a_g g}$,
and call Eq.\eqref{eq bar map} the ``bar'' map of $FG$ for convenience.
So, $\overline{\overline a}=a$, {$\overline{a b}=\overline b \overline a$},
for $a,b\in FG$. It is an automorphism of $FG$
once $G$ is abelian. The following is a linear form of $FG$:
\begin{align} \label{eq linear form FG}
\sigma:~ FG\longrightarrow F,~~
\sum\limits_{g\in G}a_g g\longmapsto a_{1_G} ~~~
(\mbox{$1_G$ is the identity of $G$}).
\end{align}

For $a=\sum_{g\in G}a_gg, b=\sum_{g\in G}b_gg\in{FG}$,
the inner product $\langle a,b\rangle=\sum_{g\in G}a_gb_g$.
The following is just \cite[Lemma II.4]{FL20}.

\begin{lemma}\label{dual-C}
{\bf(1)} $\sigma(a b)=\sigma(b a)$, $\forall~a,b\in FG$.

{\bf(2)} $\langle a, b\rangle=\sigma(a \overline b)=\sigma(\overline a b)$,
          $\forall~a,b\in FG$.

{\bf(3)} $\langle d a,\,b\rangle=\langle a,\,\overline d b\rangle$,
          $\forall~a, b,d\in FG$.

{\bf(4)} If $C$ is an $FG$-code, then so is $C^\bot$.

{\bf(5)} For $FG$-codes $C$ and $D$,
 $\langle C,D\rangle=0$ if and only if $C\overline D=0$.
\end{lemma}

Let $H$ be a cyclic group of odd order $n>1$ with $\gcd(n,q)=1$.
Let $e_{0}, e_{1}, \cdots, e_\ell$ be
all primitive idempotents of $FH$,
where $e_{0}=\frac{1}{n}\sum_{x\in{H}}x$.
Thus, $FH$ is a semisimple algebra, i.e.,
$FH$ is the direct sum of simple ideals as follows:
\begin{align} \label{eq FH=...}
FH=FH e_0\oplus FH e_1\oplus\cdots\oplus FH e_\ell,
\end{align}
where $FH e_i$ being a field with identity $e_i$.

Since the bar map Eq.\eqref{eq bar map} is an automorphism of $FH$ of order $2$,
it permutes the primitive idempotents $e_0,e_1,\cdots,e_\ell$ in Eq.\eqref{eq FH=...}
(note that $\overline e_0=e_0$).
By Lemma \ref{dual-C}(5) we have:
\begin{align}
\big\langle FH e_i, FHe_j\big\rangle=
\begin{cases} 0, & \mbox{if~ } e_i\ne\overline e_j;\\ F, & \mbox{if~ } e_i=\overline e_j.
\end{cases}
\end{align}

For any ring $R$ (with identity $1_R$), by $R^\times$ we denote the unit group, i.e.,
the multiplicative group of all the units (invertible elements) of $R$.
For the field $F$, $F^\times=F\backslash\{0\}$.
By $\mathbb Z_{n}$ we denote the integer residue ring modulo $n$,
hence $\mathbb Z_{n}^{\times}$ is the multiplicative group consisting of the reduced
residue classes. Then $q\in{\mathbb Z_{n}^{\times}}$ (since $\gcd(n,q)=1$).
In the multiplicative group $\mathbb Z_{n}^{\times}$,
${\rm ord}_{{\Bbb Z}_n^\times}(q)$ denotes the order of $q$,
and ${\langle q\rangle}_{\mathbb{Z}_{n}^{\times}}$ denotes
the cyclic subgroup generated by $q$.
The following facts are well-known.

\begin{lemma}\label{a^T=bar a}
Keep the notation be as above in Eq.\eqref{eq FH=...}.

{\bf (1)} {\rm (\cite[Theorem 1]{KR})}~ 
$\overline{e_{j}}=e_{j}$, $\forall j\ge 0$,
if and only if $-1\in{\langle q\rangle}_{\mathbb{Z}_{n}^{\times}}$.

{\bf (2)} {\rm (\cite[Theorem 6]{AKS})}~
$\overline{e_{j}}\ne e_{j}$, $\forall j>0$,
if and only if ${\rm ord}_{{\Bbb Z}_n^\times}(q)$ is odd.
\end{lemma}

Let $e_0, e_1,\cdots, e_\ell$, where $e_0=\frac{1}{n}\sum_{x\in H}x$,
be the all primitive idempotents of $FH$ as in Eq.\eqref{eq FH=...}.
Denote
\begin{align}\label{eq def lambda}
 \lambda(n)=\min\big\{\dim_{F}(FHe_{1}),\cdots, \dim_{F}(FHe_{\ell})\big\}.
\end{align}
It is known (cf. \cite[Lemma II.2]{LF22}) that:

\begin{lemma} \label{lem lambda=min}
$\lambda(n)=\min\big\{{\rm ord}_{{\Bbb Z}_p^\times}(q)\,\big|\;
\mbox{$p$ runs over the prime divisors of $n$}\big\}$.
\end{lemma}

%


For an equation of $X,Y$ over $F$, we'll need the following result.
\begin{lemma} \label{lem x^2+gxy+y^2}
Let $g\in F$ and $g\ne\pm 2$. Then the equation $X^2+gXY+Y^2=-1$
has a solution $(s,s')\in F\times F$; and, there is a solution
$(s,s')$ such that $s'=0$ if and only if
either $q$ is even or $4\,|\,(q-1)$.
\end{lemma}

\begin{proof}
Let ${\cal Q}=\{a^2\,|\,a\in F\}$ be a subset of $F$.
If $q$ is even, then ${\cal Q}=F$ and
the lemma holds obviously. Assume that $q$ is odd.
Then $|{\cal Q}|=\frac{q+1}{2}$, and
$$
 X^2+gXY+Y^2=\big(X+\frac{g}{2}Y\big)^2
 +\big(1-\frac{g^2}{4}\big)Y^2=X'^2+bY^2,
$$
where $X'=X+\frac{g}{2}Y$ and $b=1-\frac{g^2}{4}\ne 0$
(as $g\ne\pm 2$). Then $\big|(-1-b{\cal Q})\big|=\big|{\cal Q}\big|$, and
${\big|{\cal Q}\big|+\big|(-1-b{\cal Q})\big|}=q+1>|F|$.
So ${\cal Q}\cap (-1-b{\cal Q})\ne\emptyset$, and
there are $s',t'\in F$ such that $X' = t'$ and $Y=s'$ satisfying $t'^2=-1-bs'^2$.
Thus $X=s=t'-\frac{gs'}{2}$ and $Y=s'$ are a solution of the equation.
There is a solution $(s,s')$ such that $s'=0$
if and only if $-1$ is a square of $F$, so
the second conclusion is obvious.
\end{proof}

By ${\rm M}_2(F)$ we denote the
$F$-algebra consisting of all $F$-matrices of degree $2$.

\begin{lemma} \label{lem matrix size 2}
Let $M={\rm M}_2(F)$, and  $\vph(X)=X^2+gX+1$ be
an irreducible polynomial over $F$. Then there is a
subalgebra $E$ of $M$ such that $E\cong F[X]/\langle\vph(X)\rangle$,
hence~$E$ is an extension field over~$F$ of degree $2$; and the following hold.

{\bf(1)} For any $f\in M$ with ${\rm rank}(f)=1$, $Ef=Mf=:L$
is a simple left ideal of $M$
and, for $0\ne c\in L$, $a,b\in E^\times$, $ac=cb$ if and only if $a=b\in F^\times$.

{\bf(2)} There are altogether $q+1$ simple left ideals of $M$ as follows:
$$
 Mf, \qquad f=\begin{pmatrix} a&1\\ 0&0 \end{pmatrix}, \forall\, a\in F ;
 ~~ \mbox{or}~~ f=\begin{pmatrix} 1 & 0\\ 0&0 \end{pmatrix}.
$$

{\bf(3)} Let $L$ be a simple left ideal of $M$.
When $\beta$ runs over $E^\times$, $L\beta$ runs over the simple left ideals of $M$,
with each of them appears exactly $q-1$ times.
\end{lemma}

\begin{proof}
(1) and (3) have been proved in \cite[Lemma III.6]{FL20}.
We prove (2).
For $f\in M$ with ${\rm rank}(f)=1$ and $\alpha\in M$,
each row of $\alpha f$ is a linear combination of the rows of $f$.
Then the rows of $\alpha f$ with $\alpha$ running on $M$
form exactly a $1$-dimensional subspace of $F\times F$.
For the $q+1$ vectors listed in (2), i.e., $(a,1)$, $a\in F$, and $(1,0)$,
any two of them are linearly independent.
Thus each of the $q+1$ vectors generates a
$1$-dimensional subspace, and any two of the $q+1$ obtained
 $1$-dimensional subspaces are distinct.
There are altogether $q+1$ $1$-dimensional subspaces of $F\times F$.
So the $q+1$ simple left ideals listed in (2) are the all
simple left ideals of $M$.
\end{proof}

%
%

\section{Consta-dihedral group algebras}
\label{Consta-dihedral group algebras}
Keep the notation in Section \ref{preliminaries}.
In this section we characterize consta-dihedral group algebras.
From now on to the end of the paper we assume that
 (except for other explicit specified):
\begin{align} \label{eq G dihedral}
G=\langle u,v\,|\,u^n=1=v^2, vuv^{-1}=u^{-1}\rangle, \quad
  \mbox{$n>1$ is odd}, \quad\gcd(n,q)=1;
\end{align}
i.e., $G$ is the dihedral group of order $2n$; and denote
$$H=\langle u\rangle, \quad
 T=\langle v\rangle, \quad
\mbox{then}~~ G=H\rtimes T.
$$
We further consider the group
$$\tilde G=\langle u,\dot v\,|\,u^n=1=\dot v^4, \dot vu\dot v^{-1}=u^{-1}\rangle=H\rtimes\tilde T,
$$
which is a semidirect product of the cyclic group $H=\langle u\rangle$ of order $n$
by the cyclic group $\tilde T=\langle\dot v\rangle$ of order $4$
with relation  $\dot vu\dot v^{-1}=u^{-1}$.
Obviously, $Z=\{1,\dot v^2\}$ is a central subgroup of $\tilde G$,
and the quotient group $\tilde G/Z\cong G$ is the usual dihedral group of order $2n$.
The $\tilde G$ is called a dicyclic group in literature, e.g., \cite{BR}.

\begin{definition}\label{def consta-dihedral alg} \rm
Let $\cda$ be the $F$-vector space with basis
\begin{align} \label{eq F*G basis}
\big\{u^i\dot v^j\,\big|\,0\le i<n,\,0\le j<2\big\}
=\{1,u,\cdots,u^{n-1}, \dot v, \;u\dot v,\cdots,u^{n-1}\dot v\}
\end{align}
and endowed with the $F$-linear multiplication
induced by the multiplication of~$\tilde G$
with identifying that $\dot v^2=-1\in F$, i.e., subject to the following relations:
\begin{align}\label{eq def F*G}
 u^n=1,\quad \dot v^2=-1, \quad \dot v u=u^{-1}\dot v.
\end{align}
The $F$-algebra
$\cda=\big\{\sum_{h\in H}a_h h +\!\sum_{h\in H} a_{h\dot v} h\dot v
 \,\big|\, a_h, a_{h\dot v}\in F\big\}$, i.e.,
$\cda=\big\{\sum_{j=0}^1\sum_{i=0}^{n-1}a_{ij}u^i\dot v^j
 \,\big|\, a_{ij}\in F\big\}$,
is called the {\em consta-dihedral group algebra},
and any left ideal of $\cda$ is called a
{\em consta-dihedral code}   (cf. \cite{SR}).
\end{definition}

In another notation, the consta-dihedral group algebra
$$\cda =F[X,Y]\big/\langle X^n-1,\, Y^2+1,\, XYX-Y\rangle,$$
where $F[X,Y]$ is the non-commutative $F$-polynomial algebra of $X$ and $Y$,
and $\langle X^n\!-\!1, \,Y^2+1,\, XYX-Y\rangle$
is the ideal generated by $X^n\!-\!1, \, Y^2+1, \, XYX-Y$.
Therefore, $\cda $ is identified with the quotient algebra
of the group algebra $F\tilde G$ over the ideal $\langle \dot v^2+1\rangle$
generated by $\dot v^2+1$:
\begin{align} \label{eq F*G=quotient}
 \cda =F\tilde G\big/\langle \dot v^2+1\rangle.
\end{align}

\begin{remark}\rm
For any finite group  $G$,
by a general theory (\cite[p.268]{CR}), 
a function $\alpha: G \times G \rightarrow F^{\times}$
 is called a {\em $2$-cocycle} of G if
\begin{align} \label{eq cocycle}
\alpha(g_{1},g_{2}g_{3})\alpha(g_{2},g_{3})
 =\alpha(g_{1}g_{2}, g_{3})\alpha(g_{1}, g_{2}),
  \quad\forall\, g_{1},g_2,g_3\in{G}.
\end{align}
For a $2$-cocycle $\alpha$,
the {\em twisted group algebra} of $G$ by $\alpha$, denoted by $F^{\alpha}G$,
is the $F$-vector space with basis $G$ and endowed
with the $F$-bilinear product $FG \times FG \rightarrow FG$ defined by
$$
 g_{1}\cdot g_{2}=\alpha(g_{1},g_{2})(g_{1}g_{2}),
 \qquad \forall \, g_{1}, g_{2}\in{G};
$$
(the associativity of the multiplication follows from Eq.\eqref{eq cocycle}).
Turn back to the dihedral group $G=\langle u,v\,|\,u^n=1=v^2,uvu^{-1}=u^{-1}\rangle$.
It is easy to check that the following $\gamma$ is a $2$-cocycle of $G$:
\begin{equation}\label{ui-vi}
\gamma(u^{i}v^{s}, u^{j}v^{t})=
 \begin{cases}
      -1, \ & s=t=1;\\
      1,\  & \mbox{otherwise};
 \end{cases}
  \quad 0\leq i, j< n,~0\le s,t< 2;
 \end{equation}
and the consta-dihedral group algebra $\cda$
defined above (Definition~\ref{def consta-dihedral alg})
is just the twisted group algebra $F^\gamma G$
by the above $2$-cocycle $\gamma$.
\end{remark}

\begin{remark}\label{rk bar map cda} \rm
By Eq.\eqref{eq bar map}, we have the bar map on the group algebra $F\tilde G$:
$$
F\tilde G\to F\tilde G,~~
\sum_{g\in\tilde G}a_g g \mapsto\overline{\sum_{g\in\tilde G}a_g g}
=\sum_{g\in\tilde G}a_g g^{-1}.
$$
Since $\overline{\dot v^2}=\dot v^2$ (as $\dot v^4=1$),
the ideal $\langle\dot v^2+1\rangle$ of $F\tilde G$ is invariant by the bar map
(i.e., $\overline{\langle\dot v^2+1\rangle}=\langle\dot v^2+1\rangle$),
so the bar map of $F\tilde G$ induces a transformation
(called and denoted by ``bar map'' again) of the quotient algebra
$\cda=F\tilde G/\langle \dot v^2+1\rangle$ (see Eq.\eqref{eq F*G=quotient})
as follows:
for
$\sum_{h\in H}a_h h+\!\sum_{h\in H} a_{h\dot v} h\dot v\in\cda$,
\begin{align} \label{eq bar F*G}
\overline{
 \sum\limits_{h\in H}a_h h +\!\sum\limits_{h\in H} a_{h\dot v} h\dot v}
 = \sum\limits_{h\in H}a_h h^{-1}
    +\!\sum\limits_{h\in H} a_{h\dot v} (h\dot v)^{-1} \in{\cda}.
\end{align} \nopagebreak
Obviously, $\overline{\dot v} =\dot v^{-1}= -\dot v$; and
 $\dot v\,a=\overline a\,\dot v$, $\forall\,a\in FH$.
Eq.\eqref{eq bar F*G} can be rewritten in a linear combination
of the standard basis of $\cda$ in Eq.\eqref{eq F*G basis}:
\begin{align*}
\overline{\sum\limits_{h\in H}a_h h +\!\sum\limits_{h\in H} a_{h\dot v} h\dot v}
=\sum\limits_{h\in H}a_h h^{-1} -\!\sum\limits_{h\in H} a_{h\dot v} h\dot v
\,\in\,{\cda}.
\end{align*}
Because the bar map on $\cda$ is induced by the bar map on $F\tilde G$
which is an anti-automorphism,
the bar map on $\cda$ is again an anti-automorphism:
\begin{align*}
\overline{\overline a}=a, ~~ \overline{ab}=\bar b\,\bar a, ~~~
 \forall\,a,b\in \cda.
\end{align*}
\end{remark}

\begin{remark}\label{rk sigma cda} \rm
Similarly to Eq.\eqref{eq linear form FG},
with the basis Eq.\eqref{eq F*G basis} we get the map
$$
 \sigma:~ \cda \to F,~ \sum_{0\le i<n,\, 0\le j<2}a_{ij}u^i\dot v^j\mapsto a_{00},
$$
which is a linear form of~$\cda$. And we have that:
\begin{itemize}
\item
 Lemma~\ref{dual-C} is still valid
for the consta-dihedral group algebra~$\cda$.
\end{itemize}

\noindent
The proof is similar to \cite[Lemma II.4]{FL20}.
For any elements $a,b\in{\cda}$,
\begin{align*}
a\bar b&=\sum_{0\le i<n,\, 0\le j<2}a_{ij}u^i\dot v^j
\cdot\!\sum_{0\le i'<n,\, 0\le j'<2}b_{i'j'}\overline{u^{i'}\dot v^{j'}} \\
&=\sum_{0\le i,i'<n,\, 0\le j,j'<2}a_{ij}b_{i'j'}u^i\dot v^j\dot v^{-j'}u^{-i'}.
\end{align*}
Rewriting it as a linear combination of the basis Eq.\eqref{eq F*G basis} and
picking up the coefficient of $u^0\dot v^0=1_{\tilde G}$,
 we get  
\begin{align*}
\sigma(a\bar b)=\sum_{(iji'j')} a_{ij}b_{i'j'}s_{iji'j'}, 
\end{align*}
where the sum is over the indexes $(iji'j')$ 
satisfying the following two:

(i)~ $0\le i,i'<n$ and $0\le j,j'<2$;

(ii)~ $u^i\dot v^j\dot v^{-j'}u^{-i'}=s_{iji'j'}\cdot 1_{\tilde G}$ for an $s_{iji'j'}\in F$.

\noindent
By (i), it is easy to see that (ii) holds
 if and only if $i=i'$ and $j=j'$;
and at that case $s_{iji'j'}=1$ 
(note that (i) is necessary for the conclusion;
e.g, $\dot  v^3\cdot \dot v^{-1}=-1\cdot 1_{\tilde G}$ in $\cda$ but $3\ne 1$).
So $\sigma(a\bar b)=\sum_{j=0}^1\sum_{i=0}^{n-1}a_{ij}b_{ij}$; i.e., 
 (1) of Lemma~\ref{dual-C} holds for $\cda$.
In a similar way, (5) of Lemma~\ref{dual-C} for $\cda$ holds.
And, (2), (3) and (4) of Lemma~\ref{dual-C} for $\cda$ can be checked by (1) directly.
In particular, if $C$ is a consta-dihedral code,
then so is the orthogonal code $C^\bot$.
\end{remark}

By Definition~\ref{def consta-dihedral alg}, the cyclic group algebra $FH$ is a
commutative subalgebra of the consta-dihedral group algebra $\cda$;
and, as $FH$-modules, we have
\begin{align}\label{eq FH oplus ...}
\cda=FH \oplus FH\dot v=\big\{a+a'\dot v\,\big|\,a,a'\in FH\big\}.
\end{align}

\begin{lemma}\label{F-e0}
Let $FH=FHe_0\oplus FH e_1\oplus\cdots\oplus FHe_\ell$ as in Eq.\eqref{eq FH=...}.
Then the idempotent $e_{0}$ is central in $\cda$
and the ideal $\cda e_{0}$ is a commutative $F$-algebra of dimension $2$,
and the following hold.

{\bf(1)} If $q$ is odd and $4\nmid(q-1)$, then $\cda e_0$ is a field extension over $F$
with degree $|\cda e_0:F|=2$.

{\bf(2)} If either $q$ is even or $4\,|\,(q-1)$, then there is an element $r\in F$ such that $r^2=-1$
and $C_0=\cda(re_0+e_0\dot v)$ is an $1$-dimensional ideal of $\cda e_0$,
and $\langle C_0,C_0\rangle=0$.
\end{lemma}

\begin{proof}
Since $\overline{e_0}=e_0$,
$\dot ve_0=\overline{e_0}\,\dot v=e_0\dot v$.
So $e_0$ is a central element of $\cda$.
It is known that $\cda e_0=FH e_0\oplus FH e_0\dot v$
and $FH e_0=\{ae_0\,|\, a\in F\}\cong F$.
Thus $\cda e_0$ is a commutative $F$-algebra
with $e_0$, $e_0\dot v$ being a basis.

(1) Since the group $F^\times$ is a cyclic group
having no element of order $4$,
the polynomial $X^2+1$ is irreducible over $F$.
Because $(e_0\dot v)^2=-e_0$, we have an isomorphism
$\cda e_0\cong F[X]/\langle X^2+1\rangle$ which is
a field extension over $F$ of degree~$2$.

(2) If $q$ is even, then $-1=1$ and $r=1$ satisfies that $r^2=-1$.
If $4\,|(q-1)$, then $F^\times$ has an element $r$ of order $4$,
and so $r^2=-1$. Thus
$$
 \dot v(re_0+e_0\dot v)=re_0\dot v+\dot v e_0\dot v
 =re_0\dot v-e_0=re_0\dot v+r^2e_0=r(re_0+e_0\dot v).
$$
So $\dim_F (\cda(re_0+e_0\dot v))=1$. And
\begin{align*}
&(re_0+e_0\dot v)\overline{(re_0+e_0\dot v)}
=(re_0+e_0\dot v)(re_0+\bar{\dot v} e_0)\\
&=(re_0+e_0\dot v)(re_0- e_0\dot v)
=(re_0)^2-(e_0\dot v)^2 = -e_0+e_0=0.
\end{align*}
By Remark~\ref{rk sigma cda} and Lemma~\ref{dual-C}(5), $\langle C_0,C_0\rangle=0$.
\end{proof}


\begin{lemma}\label{A-M bar e not e}
Keep the notation in Eq.\eqref{eq FH oplus ...} and Eq.\eqref{eq FH=...}.
Let e be a primitive idempotent of $FH$ other than $e_0$ with $\overline e\ne e$.
Then $\tilde F:=FHe$ is a field extension over F,  $e+\overline e$
is a primitive central idempotent of $\cda$ and:

{\bf(1)} The ideal
 $\cda(e+\overline{e})
=F He\oplus FH \bar e\oplus FH e\dot v\oplus FH\bar e\dot v
\cong{\rm M}_2(\tilde F). $

{\bf(2)} With the isomorphism in (1),
if $f\in\cda(e+\overline{e})$ corresponds to the matrix
$\begin{pmatrix}a_{11}& a_{12}\\ a_{21} & a_{22}\end{pmatrix}
\in{\rm M}_2(\tilde F)$, then
$\overline f$ corresponds to the matrix
$\begin{pmatrix}a_{22}& -a_{12}\\ -a_{21} & a_{11}\end{pmatrix}$.
\end{lemma}

\begin{proof}
Since $e\bar{e}=0$, $e+\bar{e}$ is an idempotent of $\cda$.
For $\dot v\in{\cda}$, we have
$\dot v (e+\bar e) =\dot v e + \dot v\bar e=\bar e \dot v +e\dot v=(\bar e+e) \dot v $ .
Thus, $e+\bar e$ is a central element of~$\cda$.
So, $\cda (e+\bar e)=F He\oplus FH e\dot v \oplus FH \bar e\oplus FH\bar e\dot v$
is an ideal of $\cda$.
We first show an $\tilde F$-algebra isomorphism. 
Define a map:
\begin{align}\label{eq Corre 1}
\begin{array}{ccc}
{\rm M}_2(\tilde F) & \mathop{\longrightarrow}\limits^{\cong} & 
F He\oplus FH e\dot v \oplus FH \bar e\oplus FH\bar e\dot v,
\\[3pt]
\begin{pmatrix}a_{11} & a_{12}\\ a_{21} & a_{22} \end{pmatrix}
& \longmapsto &
a_{11}e - a_{12}\,e\,\dot v + \overline{a_{21}}\,\bar e\,\dot v + \overline{a_{22}}\,\bar e,
\end{array}
\end{align}
which is a linear isomorphism.
For $a_{ij},b_{ij}\in \widetilde F$, $1\le i,j\le 2$,
noting that $\dot v a_{ij}=\overline{a_{ij}}\,\dot v$ and $\dot v\dot v=-1$, we have
\begin{align*}
&\big(a_{11}e - a_{12}e\dot v + \overline{a_{21}}\,\bar e\dot v
 + \overline{a_{22}}~\overline{e}\big)
\big(b_{11}e - b_{12}e\dot v + \overline{b_{21}}\,\bar e \dot v
 +\overline{b_{22}}~\overline{e})\\
&=\;
(a_{11}b_{11}+a_{12}b_{21})e - (a_{11}b_{12}+a_{12}b_{22})e\dot v
\\
&\quad + \overline{(a_{21}b_{11}+a_{22}b_{21})}~\overline e \dot v
 +\overline{(a_{21}b_{12}+a_{22}b_{22})}~\overline e .
\end{align*}
So, Eq.\eqref{eq Corre 1} is an $\tilde F$-algebra isomorphism, and (1) holds.

Next, we have the bar map image of
$a_{11}e - a_{12}e\dot v +\overline{a_{21}}\,\bar e\dot v+\overline{a_{22}}\,\bar e$
(note that $\overline{\dot v}={\dot v}^{-1}=-\dot v$) as follows:
\begin{align*}
\overline{a_{11}e - a_{12}e\dot v +
  \overline{a_{21}}\,\bar e\dot v + \overline{a_{22}}\,\bar e}
&=\overline{a_{11}}\,\bar e - \overline{\dot v}\,\overline{a_{12}}\,\bar e
  + \overline{\dot v}\,a_{21}\, e + a_{22}\, e.\\
&= a_{22}\, e +\,a_{12}\, e\,\dot v
  - \overline{a_{21}}\,\bar e\, \dot v + \overline{a_{11}}\,\bar e.
\end{align*}
Thus, this image corresponds the matrix:
\begin{align} \label{bar to matrix}
\overline{a_{11}e - a_{12}e\dot v +
  \overline{a_{21}}\,\bar e\dot v + \overline{a_{22}}\,\bar e}
 ~\longleftrightarrow~
 \begin{pmatrix}a_{22} & -a_{12}\\ -a_{21} & a_{11} \end{pmatrix}.
\end{align}
We are done.
\end{proof}

\begin{lemma}\label{A-M bar e=e}
Let e be a primitive idempotent of FH with $\bar e=e\ne e_0$.
Then $FHe$ is a field extension over F,
$e$ is a primitive central idempotent of $\cda$,
$\tilde F :=\{a\,|\,a\in{FHe},\, a=\bar a\}$
is a subfield of $FHe$ with degree $|FHe:\tilde F|=2$,
the ideal $\cda e = FHe\oplus FHe v \cong {\rm M}_{2}(\tilde F)$,
and the center ${\rm Z}(\cda e )=\tilde F$.
\end{lemma}

\begin{proof}
Since $\dot v e =\bar e \dot v=e \dot v$ (as $e=\bar e$),
$e$ is a primitive central idempotent of $\cda$ .
The $FHe$ is a field with identity e.
Since $n>1$ is odd,
$\tilde F$ is a subfield of $FHe$
and $|FH e:\tilde F|=2$ (cf. \cite[Lemma II.3]{LF22}).
Since $FHe=\sum_{i=0}^{n-1}Fu^{i}e=\sum_{i=0}^{n-1}F(ue)^i$,
$FHe=\tilde F\oplus\tilde F(ue)$ is an extension over $\tilde F$ by the element $ue$.
And, the minimal polynomial of $ue$ over $\tilde F$ is $\vph_{ue}(X)=X^2+gX+1$,
where $\pm 2\ne g\in\tilde F$ such that
$g$ and $2$ cannot be both zero in $\tilde F$
(because $\vph_{ue}(X)$ is irreducible); cf. \cite[Lemma III.3]{FL20}.
By Lemma \ref{lem x^2+gxy+y^2},
we take $s,s'\in\tilde F$ such that $s^2+gss'+s'^2=-1$, and set
\begin{equation}\label{E-E-V}
\vep=\begin{pmatrix}
1 &  0 \\
0 &  1
\end{pmatrix},
\quad
\eta=\begin{pmatrix}
-g & 1 \\
-1&  0
\end{pmatrix},
\quad
\nu=
\begin{pmatrix}
s & s' \\
sg+s' & -s \\
\end{pmatrix}.
\end{equation}
Then the characteristic polynomial of $\eta$ is
$\varphi_{\eta}(X)=X^{2}+gX+1=\varphi_{ue}(X)$, and
$\nu^2=-\vep$ and $\nu\eta \nu^{-1}=\eta^{-1}$.
Mapping $e\mapsto \varepsilon$, $ue\mapsto \eta$ and $\dot v e\mapsto \nu$, we get
$$\cda e=FHe\oplus FH\dot ve
=\tilde F\oplus\tilde Fue\oplus\tilde F\dot v e\oplus\tilde F u\dot ve $$
and
${\rm M}_2(\tilde F)=
\tilde F\vep \oplus \tilde F\eta\oplus \tilde F\nu\oplus\tilde F\eta\nu$,
and the following (where $a,b,c,d\in{\tilde F}$)
\begin{equation}\label{eq Corre 2}
\begin{array}{ccc}
\cda e & \longrightarrow & {\rm M}_{2}(\tilde F),\\
ae+bue+c\dot v e +du \dot v e &\longmapsto & a\varepsilon+b\eta+c\nu+d\eta\nu,
\end{array}
\end{equation}
is an algebra isomorphism. Any $ae\in\tilde F$ is mapped to $a\varepsilon$;
so $\tilde F={\rm Z}(\cda e )$.
\end{proof}

\begin{theorem}\label{Main01}
The consta-dihedral group algebra $\cda$ is an orthogonal direct sum of ideals $A_t$:
\begin{align} \label{eq cda=A_0+...}
 \cda = A_{0}\oplus A_{1} \oplus \cdots\oplus A_{m},
\end{align}
where $A_{0}=\cda e_{0}$ is described in Lemma \ref{F-e0}
and, for $1\le t\le m$, the ideal
$A_{t}\cong M_{2}(F_{t})$ with $ F_{t}$ is
a field extension over F with
$\dim_{F}F_{t} =k_{t}$ and one of the following two holds:

{\bf(1)} The identity $1_{A_t}=e+\bar e$ for a primitive idempotent $e$ of $FH$
with $\bar e\ne e$, and $A_t=\cda(e+\bar e)\cong{\rm M}_2( F_t)$
as in Eq.\eqref{eq Corre 1}, where $F_t=FHe$.

{\bf(2)} The identity  $1_{A_t}=e$ for a primitive idempotent $e$ of $FH$
with $\bar e=e\ne e_0$, and $A_t=\cda e\cong{\rm M}_2( F_t)$
as in Eq.\eqref{eq Corre 2}, where $F_t=\{a\,|\,a\in{FHe},\, a=\bar a\}$
is the subfield of the field $FHe$ with degree $|FHe:F_t|=2$.
\end{theorem}

\begin{proof}
By Lemma~\ref{F-e0},  Lemma~\ref{A-M bar e not e} and Lemma~\ref{A-M bar e=e},
the direct sum in Eq.\eqref{eq cda=A_0+...}
follows at once; and for $1\le t\le m$, either (1) or (2) holds.
If $0\le s\ne t\le m$,
$$
 A_s\cdot\overline{A_t}=\cda 1_{A_s}\!\cdot\overline{\cda 1_{A_t}}
=\cda 1_{A_s}\!\cdot\overline{1_{A_t}}\cda
=\cda 1_{A_s} 1_{A_t}\cda =0.
$$
By Remark~\ref{rk sigma cda} and Lemma~\ref{dual-C}(5),
$\langle A_s, A_t\rangle=0$, $0\le s\ne t\le m$. That is,
Eq.\eqref{eq cda=A_0+...} is an orthogonal decomposition.
\end{proof}

\begin{corollary} \label{cor C=C cap A+...}
Keep the notation in Theorem~\ref{Main01}.
Let $C$, $D$ be any left ideals of $\cda$. Then:

{\bf(1)} $C=C_0\oplus C_1\oplus\cdots\oplus C_m$,
where $C_t=C\cap A_t=1_{A_t}\!\cdot\! C$, $t=0,1,\cdots,m$.
(We call $C_t$ the {\em $A_t$-component} of $C$.)

{\bf(2)} Let $D=D_0\oplus D_1\oplus\cdots\oplus D_m$
be as in~(1). Then for $c=c_0+c_1+\cdots+c_m\in C$ and
$d=d_0+d_1+\cdots+d_m\in D$, the inner product
$\langle c,d\rangle=\sum_{t=0}^m\langle c_t,d_t\rangle$.

{\bf(3)} $C^\bot=C_0^{\bot_{\!A_0}}\oplus C_1^{\bot_{\!A_1}}
  \oplus\cdots\oplus C_m^{\bot_{\!A_m}}$,
where $C_t^{\bot_{\!A_t}}$, $0\le t\le m$, denotes the orthogonal
subspace of $C_t$ in $A_t$.
Both $C_t$ and $C_t^{\bot_{\!A_t}}$ are left ideals of $A_t$.

{\bf(4)} $C$ is self-orthogonal if and only if every $C_t$, $0\le t\le m$,
is self-orthogonal in $A_t$.

{\bf(5)} $C$ is LCD if and only if every $C_t$, $0\le t\le m$, is LCD in $A_t$.
\end{corollary}

\begin{proof}
(1).
By Eq.\eqref{eq cda=A_0+...}, $1=1_{A_0}+1_{A_1}+\cdots+1_{A_m}$.
Since $C_t\subseteq C$, we have
$C_0\oplus C_1\oplus\cdots\oplus C_m\subseteq C$.
On the other hand, for $c\in C$,
$$c=1\!\cdot\! c=1_{\!A_0}\!\cdot\! c
 +1_{\!A_1}\!\cdot\! c+\cdots+1_{\!A_m}\!\cdot\! c
\in (C\cap A_{0})\oplus (C\cap A_{1})\oplus\cdots \oplus (C\cap A_{m});
$$
so $C\subseteq C_0\oplus C_1\oplus\cdots\oplus C_m$.
We call $c_t=1_{\!A_t}\!\cdot\! c$ the {\em $A_t$-component} of $c$.

(2). By Remark~\ref{rk sigma cda} and Lemma~\ref{dual-C}(1), we have
\begin{equation*}
  \begin{split}
\langle c, d\rangle
&=\sigma(c \overline d)=\sigma(c_{0}\bar d_{0}+c_{1}\bar d_{1}\cdots +c_{m}\bar d_{m})\\
&=\sigma(c_{0}\bar d_{0})+\sigma(c_{1}\bar d_{1})+\cdots+\sigma(c_{m}\bar d_{m})\\
&=\langle c_0,d_0\rangle+\langle c_1,d_1\rangle+\cdots+\langle c_m,d_m\rangle.
\end{split}
\end{equation*}

(3). By~(1), 
$C^{\bot}= C_{0}'\oplus C_{1}'\oplus \cdots \oplus C_{m}'$,
where $C_{t}'=C^\bot\!\cap\! A_{t}$.
Since $C_t=C\cap A_t\subseteq C$ and $C^{\bot}\cap A_{t}\subseteq C^\bot$,
$\langle C_t,\,C_t'\rangle\subseteq\langle C,\,C^\bot\rangle =0$.
Thus $C_t'\subseteq C_t^{\bot_{A_t}}$.
Conversely, assume that $a_t\in C_t^{\bot_{A_t}}$,
then the $A_j$-component of $a_t$ is zero provided $j\ne t$;
 for any $c=c_0+c_1+\cdots+c_m\in C$,
by the above (2), $\langle c,a_t\rangle=\langle c_t, a_t\rangle=0$;
so $a_t\in C^\bot\cap A_t=C'_t$.
We get that $C_t^{\bot_{A_t}}\subseteq C'_t$.

(4).
If $\langle C,C\rangle=0$, since $C_t\subseteq C$,
we have $\langle C_t,C_t\rangle\subseteq \langle C,C\rangle=0$;
i.e., $C_t$ is self-orthogonal in $A_t$.
Conversely, if $\langle C_t,C_t\rangle=0$ for all $t=0,1,\cdots,m$, then,
by the above (2), $\langle C,C\rangle=\sum_{t=0}^m\langle C_t,C_t\rangle=0$.

(5). As $(C\cap D)\cap A_t=(C\cap A_t)\cap(D\cap A_t)$, with notation in (1) we have
\begin{align} \label{eq C cap D}
C\cap D=(C_0\cap D_0)\oplus(C_1\cap D_1)\oplus\cdots\oplus(C_m\cap D_m).
\end{align}
By (3), $C^\bot\cap A_t=C_t^{\bot_{A_t}}$.
Applying Eq.\eqref{eq C cap D} to $D=C^\bot$, we get
$$C\cap C^{\bot}= (C_{0}\cap C_{0}^{\bot_{A_{0}}}) \oplus (C_{1}\cap C_{1}^{\bot_{A_{1}}})\oplus \cdots \oplus (C_{m}\cap C_{m}^{\bot_{A_{m}}}).
$$
Therefore, $C\cap C^{\bot}=\{0\}$ if and only if
$C_{t}\cap C_{t}^{\bot_{A_{t}}}=\{0\}$ for $t=0,1,\cdots,m$.
\end{proof}

\begin{corollary}\label{Cor-k1+k2}
Keep the notation in Theorem~\ref{Main01}.

{\bf(1)} $k_{1}+k_{2}+ \cdots + k_{m}=\frac{n-1}{2}$.

{\bf(2)}  $2k_{t}\geq \lambda(n)$,  $t=1,\ldots, m$,
where $\lambda(n)$ is defined in Eq.\eqref{eq def lambda}.
\end{corollary}

\begin{proof}
By Eq.~\eqref{eq cda=A_0+...}, we have that  $2n=\dim_{F}\cda=\sum_{t=0}^{m}\dim_{F}A_{t}$,
where $\dim_{F}A_{0}=\dim_{F}\cda e_0 =2$, see Lemma \ref{F-e0}(1).
Since $A_{t}\cong  M_{2}(F_{t})$ for $t=1,\cdots,m$, $\dim_{F}A_{t}=4k_t$,
and so $2(k_{1}+k_{2}+ \cdots + k_{m})=n-1$.
The second conclusion is obvious.
\end{proof}

\section{Consta-dihedral codes}\label{Consta-dihedral codes}

Any left ideal of the consta-dihedral algebra
$\cda$ is called a {\em consta-dihedral code} over $F$ of length $2n$,
cf.~Definition~\ref{def consta-dihedral alg}.
In this section we construct some consta-dihedral codes and
investigate their algebraic properties.

\begin{remark} \label{rk def K^x} \rm
Keep the notation in Theorem~\ref{Main01}:
$\cda = A_{0}\oplus A_{1} \oplus \cdots\oplus A_{m}$,
 where $A_{0}=\cda e_{0}$ and, for $t=1,\cdots, m$,
the ideal $A_{t}\cong {\rm M}_{2}(F_{t})$ as
described in Theorem \ref{Main01}(1) and (2),
and $\dim_{F}F_{t}=k_{t}$. 
By the isomorphism $A_{t}\cong {\rm M}_{2}(F_{t})$, 
applying Lemma~\ref{lem matrix size 2} to ${\rm M}_{2}(F_{t})$,
we get a field $K_{t}\subseteq A_{t}$
corresponding the subfield (denoted by $E$ in Lemma~\ref{lem matrix size 2})
of ${\rm M}_{2}(F_{t})$ of dimension $2$ over $F_t$.
So $\dim_{F}K_{t}=2k_{t}$ because $\dim_{F} F_{t}=k_t$.
And the following hold.

\begin{itemize}
\item[(1)]
The following is a subgroup of the multiplicative
unit group $(\cda)^\times$:
\begin{align} \label{eq K^*=...}
 K^*:=\{e_0\}\times K_1^\times\cdots\times K_m^\times,
\end{align}
where $K_t^\times=K_t \backslash \{0\}$
is the multiplicative unit group of the field $K_t$.
If $C$ is a left ideal of $A_t$ and $\beta\in K_t^\times$,
then $C\beta$ is a left ideal of $A_t$ which is isomorphic to $C$
(cf. Lemma~\ref{lem matrix size 2}).
\item[(2)]
If $1_{A_t}=e$ for a primitive idempotent $e$ of $FH$ with $\bar e=e\ne e_0$,
then we choose $K_t=FH e$,
hence $F_{t}=\tilde F =\{a\in{K_t}\,|\,\bar a =a\}={\rm Z}(A_t)$
as described in Lemma~\ref{A-M bar e=e},
where ${\rm Z}(A_t)$ denotes the center of $A_t$.
\end{itemize}
\end{remark}

\begin{lemma} \label{lem C_t 1=e+bar e}
Assume that $1_{A_t}\!=e+\bar e$ for a primitive idempotent $e$ of $FH$
with $\bar e\ne e$. Let $C_t=A_t e$ and $\beta_t\in K_t^\times$.
Then $C_t\beta_t$ is a simple left ideal of $A_t$ and
 $\langle C_t\beta_t, C_t\beta_t\rangle=0$.
\end{lemma}

\begin{proof}
Let $M:={\rm M}_2(F_t)$, where $F_t=FHe$, see Theorem \ref{Main01}(1).
By Eq.\eqref{eq Corre 1},
$C_t=A_t e$ corresponds to $M\!\cdot\!\begin{pmatrix}1&0\\ 0&0\end{pmatrix}$,
which is a simple left ideal of $M$.
Hence $C_t$ and $C_t\beta_t$ are simple left ideal of $A_t$.
If $C_t\beta_t=C_t$, then
 $C_t\beta_t\cdot\overline{C_t\beta_t}=A_t e\overline{A_t e}=A_t e\overline{e}A_t=0$.
By Remark~\ref{rk sigma cda} and Lemma~\ref{dual-C}(5),
$\langle C_t\beta,C_t\beta\rangle=0$.
In the following, we assume that $C_t\beta_t\ne C_t$.
By Lemma~\ref{A-M bar e not e} and Lemma~\ref{lem matrix size 2},
the simple left ideal $C_t\beta_t$
corresponds the simple left ideal $Mf$ of $M$ as follows
$$
 C_t\beta_t ~\longleftrightarrow ~ Mf, \qquad
f=\begin{pmatrix} a& 1\\0&0\end{pmatrix},~~ a\in F_t.
$$
Let $f'\in A_t$ correspond $f\in M$ by the isomorphism Eq.\eqref{eq Corre 1}.
Then $C_t\beta_t=A_tf'$.
By Lemma~\ref{A-M bar e not e}(2) (cf. Eq.\eqref{bar to matrix}),
$\overline{f'}$ corresponds the matrix
$\begin{pmatrix} 0 & -1\\0&a\end{pmatrix}$.
Because
$$
\begin{pmatrix} a& 1\\0&0\end{pmatrix}
\begin{pmatrix} 0 & -1\\0&a\end{pmatrix}
=\begin{pmatrix}  0 & 0\\0&0\end{pmatrix},
~~\mbox{hence}~~ f'\!\cdot\!\overline{f'}=0.
$$
So $C_t\beta_t\cdot\overline{C_t\beta_t}=A_tf'\!\cdot\!\overline{f'}A_t=0$,
i.e., $\langle C_t\beta_t, C_t\beta_t\rangle=0$.
\end{proof}

\begin{lemma} \label{lem C_t 1=e}
Assume that $1_{\!A_t}\!=e$ for a primitive idempotent $e$ of $FH$
with $\bar e=e\ne e_0$,
and set $C_t=A_tf$ where $f=se-s'ue+\dot v e\in A_t$ with
$s,s'\in F_t$ as in Eq.\eqref{E-E-V}.
Then $C_t$ is a simple left ideal of $A_t$, and for any $\beta_t\in K_t^\times$,
\begin{align} \label{eq <C_t,c_t>}
 \langle C_t\beta_t,C_t\beta_t\rangle=\begin{cases}
   0, & \mbox{if $q$ is even or $4\,\big|\,(q^{k_t}-1)$;}\\
   F, &\mbox{otherwise.}\end{cases}
\end{align}
\end{lemma}

\begin{proof}
Note that in this case we choose $K_t=FHe$,
$F_t=\{a\,|\,a\in K_t, \overline a=a\}={\rm Z}(A_t)$ with $|K_t:F_t|=2$,
see Theorem \ref{Main01}(2) and Remark~\ref{rk def K^x}(2).
By Eq.\eqref{E-E-V} and its notation, $g,s,s'\in F_t$ satisfy that
 $g$ and $2$ are not both zero,
$g\ne\pm 2$ which implies that
$\det\begin{pmatrix} 2 & g\\ g & 2 \end{pmatrix}\ne 0$,
and $s^2+gss'+s'^2=-1$.
So the matrix
$$
 s\vep-s'\eta+\nu =\begin{pmatrix}2s+gs' & 0\\ gs+2s' &0 \end{pmatrix}
$$
has non-zero first column,
hence ${\rm rank}(s\vep-s'\eta+\nu)=1$.
By Lemma \ref{lem matrix size 2}(1),
${\rm M}_2(F_t)(s\vep-s'\eta+\nu)$ is a simple left ideal of
${\rm M}_2(F_t)$. By the isomorphism Eq.\eqref{eq Corre 2},
 $s\vep-s'\eta+\nu$ corresponds the element
$f=se-s'ue+\dot v e\in A_t$, and $C_t=A_tf$
is a simple left ideal of $A_t$.
For $\beta_t\in K_t^\times$,
obviously, $\overline{\beta_t\bar\beta_t}=\beta_t\bar\beta_t$;
so $\beta_t\bar\beta_t\in F_t={\rm Z}(A_t)$. Then we have
$$
C_t\beta_t\cdot \overline {C_t\beta_t}=A_tf\beta_t\cdot\overline{A_tf\beta_t}
=A_t \,f\beta_t\,\overline \beta_t\,\overline f \,\overline A_t
=A_t \,f\,\overline f\,\beta_t\,\overline \beta_t A_t.
$$
Thus $C_t\beta\cdot \overline {C_t\beta}=0$
if and only if $f\,\overline f=0$.
Since $\overline u=u^{-1}$ and $\overline{\dot v}=-\dot v$, we get
\begin{align*}
f\bar f
&= (se-s'ue+\dot v e)\!\cdot\!\overline{(se-s'ue+\dot v e)}\\
&=(se - s'u e + \dot v e)\!\cdot\! (se-s'u^{-1}e- \dot v e)\\
&=(s^2+s'^2+1)e-ss'(u^{-1}e+ue)+s'(ue-u^{-1}e)\dot ve.
\end{align*}
By Eq.\eqref{eq Corre 2}, $u^{-1}e$ corresponds the matrix
$$
\eta^{-1}=\begin{pmatrix}-g & 1 \\ -1&  0 \end{pmatrix}^{-1}
=\begin{pmatrix} 0 & -1 \\ 1&  -g \end{pmatrix}.
$$
So we have the following correspondence:
$$
u^{-1}e+ue~\longleftrightarrow~ \begin{pmatrix} -g & 0 \\ 0 & -g \end{pmatrix},
\qquad
ue-u^{-1}e~\longleftrightarrow~ \begin{pmatrix} -g & 2 \\ -2 & g \end{pmatrix}.
$$
Thus $ue-u^{-1}e$ is invertible because
$\det\begin{pmatrix} g & -2 \\ 2 & -g \end{pmatrix}=4-g^2\ne 0$.
And $(s^2+s'^2+1)e-ss'(u^{-1}e+ue)$ corresponds to the matrix
$$
(s^2+s'^2+1)\begin{pmatrix} 1 & 0 \\ 0 & 1 \end{pmatrix}
-ss'\begin{pmatrix} -g & 0 \\ 0 & -g \end{pmatrix}.
$$
Because $s^2+s'^2+1+gss'=0$, we get $(s^2+s'^2+1)e-ss'(u^{-1}e+ue)=0$.
 Therefore
\begin{align} \label{eq f bar f}
f\,\bar f=s'(ue-u^{-1}e)\dot ve, \qquad
ue\!-\!u^{-\!1}e,\, \dot ve\,\in\,A_t^\times.
\end{align}
It follows that $f\,\bar f=0$ if and only if $s'=0$.
As $s,s'\in F_t$ and $|F_t|=q^{k_t}$,
Eq.\eqref{eq <C_t,c_t>} follows from Lemma \ref{lem x^2+gxy+y^2}.
\end{proof}

\begin{remark} \label{rk notation Z C ...} \rm
In the following we always fix:

--- $C_t=A_te$ as in Lemma~\ref{lem C_t 1=e+bar e}
if it is the case of Theorem \ref{Main01}(1);

--- $C_t=A_tf$ as in Lemma~\ref{lem C_t 1=e}
if it is the case of Theorem \ref{Main01}(2);

\noindent
and consider the following consta-dihedral code
\begin{align}\label{eq.C-O}
 C=C_1\oplus\cdots\oplus C_m.
\end{align}
We have the following:

(1)~ For $C$ in Eq.(\ref{eq.C-O})
the rate ${\rm R}(C)=\frac{1}{2}-\frac{1}{2n}$,
because
\begin{align} \label{eq dim C_1+...+C_m}
\dim_F C=\sum_{t=1}^{m}\dim_F C_t=2\sum_{t=1}^{m}k_t = n-1.
\end{align}

(2)~
By $C_t^{\bot_{\!A_t}}$ we denote the orthogonal
submodule (left ideal) of $C_t$ in $A_t$ (see Corollary~\ref{cor C=C cap A+...}(3)), and
so $C_t\cap C_t^{\bot_{\!A_t}}$ is still a left ideal of $A_t$.
Since $C_t$ is simple,  $C_t\cap C_t^{\bot_{\!A_t}}$ is either $C_t$ or $0$.
Hence,
\begin{align}\label{C-cap}
C_t\cap C_t^{\bot_{\!A_t}}=
\begin{cases} C_t, & \langle C_t,C_t\rangle=0;\\ 0, & \mbox{otherwise.} \end{cases}
\end{align}
\end{remark}

\begin{lemma} \label{lem LCD case}
Assume that $q$ is odd and $4\nmid(q-1)$.
Assume that there are $m'$ indexes $1\le i_1<\cdots<i_{m'}\le m$
such that, for $t=1,\cdots,m'$,
$e_{i_t}=\overline e_{i_t}$ and $k_{i_t}$ is odd.
Let $C'=C_{i_1}\oplus\cdots\oplus C_{i_{m'}}$
and $\widehat C'=A_0\oplus C'$. Then for any $\beta\in K^*$
both $C'\beta$ and $\widehat C'\beta$ are LCD consta-dihedral codes.
\end{lemma}

\begin{proof}
Note that $q^{k_{i_t}}-1=(q-1)(q^{k_{i_t}-1}+q^{k_{i_t}-2}+\cdots+q+1)$.
Since both $k_{i_t}$ and $q$ are odd and $4\nmid(q-1)$,
we have $4\nmid(q^{k_{i_t}}-1)$.
Write $\beta=e_{0}+\beta_{1}+\cdots+\beta_{m}$,
where $\beta_{t}\in K_{t}^{\times}$ for $t=1,\cdots,m$.
Then $$C^{\prime}\beta= C_{i_1}\beta_{i_1}+\cdots+ C_{i_{m'}}\beta_{i_{m'}}.$$
By Lemma~\ref{lem C_t 1=e},  $\langle C_{i_t}\beta_{i_t}, C_{i_t}\beta_{i_t}\rangle\neq 0$
for $t=1,\cdots,m'$.
By Eq(\ref{C-cap}), we have $C_{i_t}\beta_{i_t}\cap (C_{i_t}\beta_{i_t})^{\bot_{A_{i_t}}}=0$.
Then $C_{i_t}\beta_{i_t}$ is an LCD consta-dihedral code in~$A_{i_t}$.
By Lemma \ref{cor C=C cap A+...} (5), $C'\beta$  is an LCD consta-dihedral code.
Moreover,  by Lemma~\ref{F-e0}(1),
$\langle A_{0}, A_{0}\rangle=  A_{0}\overline {A_{0}}=A_{0} \neq 0$;
so  $\widehat{C}'\beta$  is an LCD consta-dihedral code.
\end{proof}

For integers $s,t$ and a prime $p$, $p^s\,\Vert\,t$ means
that $p^s\,|\,t$ but $p^{s+1}\nmid t$.

\begin{theorem} \label{thm LCD case}
Assume that $q$ is odd and $4\nmid(q-1)$.
Assume that $-1\in{\langle q\rangle}_{\mathbb{Z}_{n}^{\times}}$
and $2\,\Vert\, {\rm ord}_{{\Bbb Z}_n^\times}(q)$.
Then for any $\beta \in K^*$, both $C\beta$
and $A_0\oplus C\beta$ are LCD consta-dihedral codes.
\end{theorem}

\begin{proof}
Since $-1\in{\langle q\rangle}_{\mathbb{Z}_{n}^{\times}}$,
by Lemma~\ref{a^T=bar a}(1), $\overline e_t=e_t$,
$t=1,\cdots,m$. By Lemma~\ref{lem LCD case},
it is enough to show that any $k_t$ is odd for $1\le t\le m$.
By Theorem~\ref{Main01}(2), $k_t=\frac{1}{2}\dim_F FH e_t$, i.e.,
$FH e_t$ is a field extension over $F$ with degree $2k_t$.
There exists a $q$-coset $Q\subseteq {\Bbb Z}_n$ such that
$\dim_F FH e_t=|Q|$. However, $|Q|$ is a divisor of
${\rm ord}_{{\Bbb Z}_n^\times}(q)$. Hence, by
the assumption that $2\,\Vert\, {\rm ord}_{{\Bbb Z}_n^\times}(q)$,
$k_t=\frac{1}{2}|Q|$ is odd.
\end{proof}

\begin{theorem} \label{thm self-dual case}
Assume that $q$ is even or $4\,\big|\,(q-1)$,
Assume that $r\in F$ satisfies that $r^2=-1$.
Set $C_0=A_0(re_0+e_0\dot v)$ as in Lemma~\ref{F-e0}(2), and
$$
\widehat C=C_0\oplus C=C_0\oplus C_1\oplus\cdots\oplus C_m.
$$
Then, for any $\beta\in K^*$, $\widehat C\beta$ is a self-dual consta-dihedral code.
\end{theorem}

\begin{proof}
By Eq.\eqref{eq K^*=...} we write $\beta=e_0+\beta_1+\cdots+\beta_m$
with $\beta_t\in K_t^\times$ for $t=1,\cdots,m$. Then
$$
\widehat C\beta=C_0\oplus C_1\beta_1\oplus\cdots\oplus C_m\beta_m,
$$
and by Corollary~\ref{cor C=C cap A+...}(2),
$$
\langle \widehat C\beta,\widehat C\beta\rangle=
\langle C_0,C_0\rangle+\langle C_1\beta_1,C_1\beta_1\rangle
+\cdots+\langle C_m\beta_m,C_m\beta_m\rangle.
$$
By Lemma~\ref{lem C_t 1=e+bar e} and Lemma~\ref{lem C_t 1=e},
$\langle C_t\beta_t,C_t\beta_t\rangle=0$, $t=1,\cdots,m$.
By Lemma~\ref{F-e0}, $\langle C_0,C_0\rangle=0$.
Hence $\widehat C\beta$ is self-orthogonal.
Finally, by Lemma~\ref{F-e0} and Eq.\eqref{eq dim C_1+...+C_m},
$\dim_F(\widehat C\beta)=\dim_F(C_0)+\dim_F(C)=n$.
So $\widehat C\beta$ is self-dual.
\end{proof}

\begin{remark}\label{rk  b...}\rm
Recall that, for any subset
 $I_{*}=\{i_1,\cdots,i_k\}\subseteq I=\{1,2,\cdots, n\}$ ($1\le i_1<\cdots<i_k\le n$),
there is the projection $\rho_{I_{*}}:$
\!$F^I\to F^{I_{*}}$,\! $(a_1,a_2\cdots,a_n)\mapsto(a_{i_1},\cdots,a_{i_k})$.
For $B\subseteq F^n$ with $|B|=q^k$, if there are subsets (repetition is allowed)
$I_1,\cdots, I_s$ of $\{1,2,\cdots,n\}$
such that: (1) for $1\le j\le s$, the projection $\rho_{I_{j}}:$
$F^n\to F^{I_j}$ maps $B$ bijectively onto $F^{I_j}$
(such $I_j$ is called an {\em information index set} of the code $B$);
(2) there is an integer $t$ such that for any $1\le i\le n$
the number of the subsets $I_j$ which contains $i$ (i.e., $i\in I_j$) equals $t$;
then $B$ is called a {\em balanced code}.
An important result (cf. \cite[Corollary 3.4]{FL15}) is that,
if $B$ is balanced, then the cardinality
$|B^{\le\delta}|\le q^{kh_q(\delta)}$ for $0\le\delta\le 1-q^{-1}$, where
\begin{align} \label{eq def B^<=}
 B^{\le\delta}=\{c\,|\, c\in B,\, {{\rm w}(c)}\le\delta n\},
\end{align}
and
\begin{align} \label{eq def h_q}
h_{q}(\delta)=
\delta \log_{q}(q-1)-\delta \log_{q}(\delta)-(1-\delta)\log_q(1-\delta),\quad
 \delta \in[0,1-q^{-1}],
\end{align}
is the {\em $q$-entropy function}.
It is easy to see that group codes are balanced, cf. \cite[Remark II.3]{FL20}.
And it is also easy to prove that constacyclic codes are balanced,
see \cite[Lemma II.8]{FL22}.
\end{remark}

\begin{lemma}\label{q01}
If $B$ is a consta-dihedral code over $F$, then
$B$ is balanced; in particular, for $0\le \delta\le 1-q^{-1}$,
the cardinality $|B^{\leq \delta}|\leq q^{\dim_F\!B \cdot h_{q}(\delta)}$.
\end{lemma}
\begin{proof}
For any $a=
\sum_{i=0}^{n-1}a_{i0}u^i\dot v^0 +\sum_{i=0}^{n-1}a_{i1}u^i\dot v^1\in\cda$,
as a word of $F^{2n}$,
\begin{align}\label{eq word a}
a=(a_{00}, a_{10}, \cdots, a_{n-1,0}, a_{01}, a_{11}, \cdots, a_{n-1,1}),
\end{align}
the coordinates of the word $a$ indexed
by the standard basis Eq.\eqref{eq F*G basis} of $\cda$:
$$
 I=\{u^{0}\dot v^{0}, u^{1}\dot v^{0}, \cdots , u^{n-1}\dot v^{0},
u^{0}\dot v^{1}, u^{1}\dot v^{1}, \cdots, u^{n-1}\dot v^{1}\}.
$$
Let ${\rm Sym}(I)$ be the symmetric group of the set $I$.
The dihedral group $G$ (as in Eq.\eqref{eq G dihedral})
acts on the set $I$ through the homomorphism
$\theta:G\to{\rm Sym}(I)$, $g\mapsto\theta_g$, as follows:
$\theta_u$ is the permutation: $\theta_u(u^i\dot v^j)=u^{i+1}\dot v^j$
(since $u\cdot u^i\dot v^j=u^{i+1}\dot v^j$), i.e.,
\begin{align} \label{eq theta_u}
 \theta_{u}=( u^{0}\dot v^{0}, u^{1}\dot v^{0}, \cdots, u^{n-1}\dot v^{0})
(u^{0}\dot v^{1}, u^{1}\dot v^{1}, \cdots, u^{n-1}\dot v^{1})
\end{align}
is a double circulant permutation of $I$; and $\theta_v$ is the permutation:
\begin{align} \label{eq theta_v}
\theta_{v}(u^i\dot v^j)
=\begin{cases} u^{n-i}\dot v, & j=0;\\ u^{n-i}, & j=1;\end{cases}
\quad\Big(\mbox{since }~
 \dot v(u^i\dot v^j)=\begin{cases} u^{n-i}\dot v, & j=0; \\
   -u^{n-i}, & j=1;\end{cases} \Big)
\end{align}
i.e.,
$$
 \theta_{v}=(u^0\dot v^0, u^0\dot v^1)(u^1\dot v^0, u^{n-1}\dot v^1)
 \cdots(u^{n-1}\dot v^0, u^1\dot v^1)
$$
is a product of $n$ transpositions of $I$.
In fact, there is a bijection $I\to G$ (by dropping the dot from $\dot v$)
 such that the action of $G$ on $I$ defined by $\theta$ as above
is equivalent to the left regular action of $G$ on $G$ itself.
In particular, $G$ acts on~$I$ transitively.
Next, let $\Theta_u$ be the permutation matrix of the permutation~$\theta_u$, i.e.,
$$
\Theta_{u}=\begin{pmatrix} Y& 0\\ 0&Y\end{pmatrix}_{2n\times2n},
~~~~ \mbox{where}~~
Y=\begin{pmatrix}
0 &0 & \cdots & 1\\ 1 & 0 & \cdots  &0\\
& \ddots & \ddots & \vdots\\
 & & 1& 0\end{pmatrix}_{n\times n}.
$$
As shown in Eq.\eqref{eq theta_v}, we should take
$$
\Theta_{\dot v}=\begin{pmatrix} 0& -X\\ X&0\end{pmatrix}_{2n\times2n},
~~~~ \mbox{where}~~
X=\begin{pmatrix}
1 &0 & \cdots & 0\\ \vdots & \vdots & \begin{sideways}$\ddots$\end{sideways}  &1\\
\vdots& 0 & \begin{sideways}$\ddots$\end{sideways} & \\
0 & 1& & \end{pmatrix}_{n\times n}.
$$
By Eq.\eqref{eq word a}, any $a\in\cda$ is identified with a word in $F^{2n}$,
i.e., a $1\times 2n$ matrix; in this way $a\Theta_u$ and $a\Theta_{\dot v}$ make sense.
By Eq.\eqref{eq theta_u} and Eq.\eqref{eq theta_v} we have
\begin{align}\label{eq a Theta}
 u a=a\Theta_u, \quad \dot v a=a\Theta_{\dot v}, \qquad \forall\,a\in\cda.
\end{align}

Let $J=\{u^{i_1}\dot v^{0}, u^{i_2}\dot v^{0}, \cdots, u^{i_s}\dot v^{0},
u^{j_1}\dot v^{1}, u^{j_2}\dot v^{1}, \cdots, u^{j_t}\dot v^{1}\}\subseteq I$,
$s+t=k$, $0\leq i_1 < i_2< \cdots< i_s \leq n-1$ and
$0\leq j_1 < j_2< \cdots< j_t \leq n-1$.
For $a=
\sum_{i=0}^{n-1}a_{i0}u^i\dot v^0 +\sum_{i=0}^{n-1}a_{i1}u^i\dot v^1\in\cda$,
by Remark \ref{rk  b...} we can write
$$
\rho_J(a)=a_{i_1 0}u^{i_1}\dot v^{0}
 + \cdots+ a_{i_s 0}u^{i_s}\dot v^{0} +
a_{j_1 1}u^{j_1}\dot v^{1}+\cdots+a_{j_t 1} u^{j_t}\dot v^{1};
$$
by Eq.\eqref{eq word a}, we can identify $\rho_J(a)$ with such a word in $F^{2n}$
whose coordinates outside~$J$ are zero. Then
$\rho_J(a)\Theta_u$ ($\rho_J(a)\Theta_{\dot v}$, resp.)
 makes sense and it is a word in $F^{2n}$
whose coordinates outside~$\theta_u(J)$ (outside~$\theta_v(J)$, resp.) are zero.
By Eq.\eqref{eq a Theta},
$\rho_J(a)\Theta_u=\rho_{\theta_u\!(\!J\!)}(ua)$ and
$\rho_J(a)\Theta_{\dot v}=\rho_{\theta_{v}\!(\!J\!)}(\dot va)$.
Replacing $a$ by $u^{-1}a$ (and replacing $a$ by $\dot v^{-1}a$)
we have
\begin{align}\label{eq rho_theta}
 \rho_{\theta_u\!(\!J\!)}(a)=\rho_J(u^{-1}a)\Theta_u, \quad
 \rho_{\theta_{v}\!(\!J\!)}(a)=\rho_J(\dot v^{-1}a)\Theta_{\dot v},\qquad
\forall\, a\in\cda.
\end{align}

Assume that $\dim_F B=k$, and $I_*\subseteq I$ is an information index set of $B$,
i.e., $|I_*|=k$ and $\rho_{I_*}(B)=F^{I_*}$.
Because $B$ is a left ideal of $\cda$, $u^{-1}B=B=\dot v^{-1} B$.
Note that both $\Theta_u$ and $\Theta_{\dot v}$ are invertible.
By Eq.\eqref{eq rho_theta},
$$
 \rho_{\theta_u\!(I_*\!)}(B)=\rho_{I_*}(u^{-1}B)\Theta_u=F^{\theta_u\!(I_*\!)},
\quad
 \rho_{\theta_{v}\!(I_*\!)}(B)=\rho_{I_*}(\dot v^{-1}B)\Theta_{\dot v}
  =F^{\theta_{v}\!(I_*\!)}.
$$
In other words, both $\theta_u(I_*)$ and $\theta_{v}(I_*)$
are information index sets of $B$.
For any $g\in G$,
because $g$ can be written as a product of several $u$ and $v$,
$\theta_g$ is a product of several $\theta_u$ and $\theta_v$;
so $\theta_g(I_*)$ is an information index set of $B$ too.

Finally, fix an information index set $I_*\subseteq I$ of $B$. Then
the $2n$ subsets (repetition allowed):
$\theta_g(I_*)$, $g\in G$, are all information index sets of $B$.
And, since $G$ acts on $I$ transitively,
by \cite[Lemma II.9]{FL22}, there is an integer $t$ such that
for any $u^i\dot v^j\in I$ the number of such
$g\in G$ that $u^i\dot v^j\in \theta_g(I)$ equals $t$.
In conclusion, $B$ is a balance code.
\end{proof}

\section{Asymptotic property of consta-dihedral codes}
\label{Asymptotic property of consta-dihedral codes}

Keep the notation in Section \ref{Consta-dihedral group algebras}
and Section \ref{Consta-dihedral codes}.
In this section we always denote
$$ A=A_1\oplus\cdots\oplus A_m, \quad
\mbox{where $A_{t}\cong {\rm M}_{2}(F_{t})$, $\dim_{F}A_t=4k_t$, $t=1,\cdots,m$.}
$$
Then $\cda=A_0\oplus A$.
By Corollary \ref{Cor-k1+k2}, we have that
\begin{align} \label{eq k_1+...}
	2k_t\ge \lambda(n),~ t=1,\cdots,m; \qquad
	k_1+\cdots+k_m=(n-1)/2.
\end{align}
We further assume that
\begin{align}\label{eq k_1<=}
	\lambda(n)/2\le k_1\le k_2\le\cdots\le k_m.
\end{align}
From now on, let $\delta$ be a real number satisfying that
($h_q(\delta)$ is defined in Eq.\eqref{eq def h_q})
\begin{align}\label{eq delta in...}
	\delta\in(0,1-q^{-1}) \quad\mbox{and}\quad
    h_q(\delta)< 1/4.
\end{align}
In this section, we prove that
the consta-dihedral codes constructed in the last section
are asymptotically good.

\subsection{Consta-dihedral codes of rate $\frac{1}{2}-\frac{1}{2n}$}

In this subsection,
we consider the consta-dihedral code $C=C_1\oplus\cdots\oplus C_m$
 defined in Eq.\eqref{eq.C-O}. Recall that
$K^*=\{e_0\}\times K_1^\times\times\cdots\times K_m^\times$
with each field $K_t\subseteq A_t$
of dimension $\dim_F K_t=2k_t$, see Remark~\ref{rk def K^x}.
For any $\beta\in K^*$, $\beta=e_0+\beta_1+\cdots+\beta_m$,
we have a consta-dihedral code
$C\beta=C_1\beta_1\oplus\cdots\oplus C_m\beta_m$.

For any $0\ne d\in A$, there is a unique subset
$\omega_d=\{t_1,\cdots,t_r\}\subseteq\{1,2,\cdots,m\}$ such that
$d=d_{t_1}+\cdots+d_{t_r}$, where
$d_{t_i}\in A_{t_i}\backslash\{0\}$ for $i=1,\cdots,r$;
we denote 
\begin{align} \label{eq ell_d<=}
\ell_d=k_{t_1}+\cdots+k_{t_r}, \quad~
\mbox{by Eq.\eqref{eq k_1+...} and Eq.\eqref{eq k_1<=},}\quad
   k_1\le \ell_d\le(n-1)/2.
\end{align}

\begin{lemma}\label{lem K_d}
Let $0\ne d\in A$.
Set ${\cal K}(C)_d=\{\beta\in K^*\;|\;d\in C\beta \}$.
Then
$$
|{\cal K}(C)_d|\le |K^*|\big/q^{\ell_d}.
$$
\end{lemma}
\begin{proof}
Assume that $\omega_d=\{t_1,\cdots,t_r\}\subseteq\{1,2,\cdots,m\}$,
and $d=d_{t_1}+\cdots+d_{t_r}$ for $d_{t_i}\in A_{t_i}\backslash\{0\}$.
Then $d\in C\beta$ if and only if
$d_{t_i}\in C_{t_i}\beta_{t_i}$, $i=1,\cdots,r$.
The $C_{t_i}\beta_{t_i}$ is a simple left ideal of $A_{t_i}$.
In $A_{t_i}$, the intersection of any two distinct simple left ideals
is~$0$; so there is at most one simple left ideal $C'_{t_i}$ containing~$d_{t_i}$.
By Lemma~\ref{lem matrix size 2}(3),
there are exactly
$q^{k_{t_i}}-1$ elements $\beta_{t_i}$ in $K_{t_i}^{\times}$ such that $C_{t_i}\beta_{t_i}=C'_{t_i}$.
 Thus
$$
\big|\{\beta_{t_i}\in K_{t_i}^\times \,|\,
  d_{t_i}\in C_{t_i}\beta_{t_i}\}\big|
\le q^{k_{t_i}}-1 .
$$
Since $\dim_{F}K_{t_i}=2k_{t_i}$, see Remark \ref{rk def K^x}, we get
$q^{k_{t_i}}-1 =\frac{|K_{t_i}^\times|}{(q^{k_{t_i}}+1)}$.
Set $\omega_d'=\{1,2,\cdots,m\}\backslash\omega_d$. Then
\begin{align*}
\big|{\cal K}(C)_d\big|
\le \prod_{t'\in\omega_d'}\! |K_{t'}^\times|\cdot
\prod_{t\in\omega_d}\! \frac{|K_t^\times|}{q^{k_{t}}+1}
= \prod_{t=1}^{m}|K_t^\times|\Big/\prod_{t\in\omega}(q^{k_{t}}+1)
\le |K^*|\Big/\prod_{t\in\omega}q^{k_{t}},
\end{align*}
i.e.,
$\big|{\cal K}(C)_d\big|\le
|K^*|\big/q^{k_{t_1}+\cdots+k_{t_r}}
=|K^*|\big/q^{\ell_d}$.
\end{proof}

Denote
\begin{align} \label{}
\Omega=\{A_{t_1}\oplus\cdots\oplus A_{t_r}~|~\{t_1,\cdots,t_r\}\subseteq\{1,\cdots,m\}\},
\end{align}
which is the set of all ideals of $A$.

\begin{lemma}\label{lem K^<delta}
Set ${\cal K}(C)^{\!\le\delta}=
\{\,\beta\in K^*\,|\,\Delta(C\beta)\le\delta\,\}$.
If $\frac{1}{4}-h_q(\delta)-\frac{\log_q n}{\lambda(n)}>0$,
then
\begin{align} \label{eq K^<=...}
\big|{\cal K}(C)^{\le\delta}\big| \le |K^{*}|\!\cdot\!
q^{-2\lambda(n)\big(\frac{1}{4}- h_q(\delta)
	-\frac{\log_q n}{\lambda(n)}\big)}.
\end{align}
\end{lemma}

\begin{proof}
For any subset $\omega\subseteq\{1,\cdots,m\}$, we denote
$A_\omega=\bigoplus_{t\in\omega}A_t$, so $ A_\omega\in{\Omega}$.
For $k_1\le\ell\le\frac{n-1}{2}$, we set
\begin{align} \label{eq A_ell=}
\begin{array}{l}
{\cal A}_\ell =\big\{\,A_\omega\in{\Omega}\;\big|\;
 \dim_{F} A_\omega=4\ell\,\big\};
\\[5pt]
{\cal D}_\ell=\big\{\,d\in A\;\big|\;
  0<{\rm w}(d)/2n\le\delta,\,\ell_d=\ell\,\big\}.
\end{array}
\end{align}
For $A_\omega\in {\cal A}_\ell$,
$\dim_{F} A_\omega=4\sum_{t\in\omega}k_t=4\ell$  and $k_t\ge k_1$,
by the assumption Eq.\eqref{eq k_1<=},
we have that $|\omega|\le \ell/k_1$. Thus,
\begin{align} \label{eq |cal A|<=}
|{\cal A}_\ell|\le m^{\ell/k_1}\le n^{\ell/k_1}.
\end{align}
It is obvious that (where
${A_\omega}^{\!\le\delta}$ is defined in Eq.\eqref{eq def B^<=})
\begin{align} \label{eq ...D_ell}
{\cal D}_\ell\subseteq \bigcup_{A_\omega\in{\cal A}_\ell}
{A_\omega}^{\!\le\delta}
\qquad\mbox{and}\qquad
{\cal K}(C)^{\!\le\delta}=
\bigcup_{\ell=k_1}^{(n-1)/2}
\bigcup_{d\in{\cal D}_\ell} {\cal K}(C)_d.
\end{align}
By Lemma~\ref{q01}, for $A_\omega\in{\cal A}_\ell$,
we have that
$|{{ A}_\omega}^{\!\le\delta}|\le q^{4\ell h_q(\delta)}$
since $\dim_{F} A_\omega=4\ell$.
By Eq.\eqref{eq ...D_ell}, we get
\begin{align*}
|{\cal D}_\ell|&\le \sum_{A_\omega\in{\cal A}_\ell}
|{{A}_\omega}^{\!\le\delta}|
\le
|{\cal A}_\ell|\cdot q^{4\ell h_q(\delta)}
\le n^{\frac{\ell}{k_1}} q^{4\ell h_q(\delta)}
=q^{4\ell h_q(\delta)+\frac{\ell\log_q n}{k_1}}.
\end{align*}
By Lemma~\ref{lem K_d}, $|{\cal K}(C)_d|\le |K^*|\big/q^{\ell}$.
From Eq.\eqref{eq ...D_ell} we obtain
\begin{align*}
|{\cal K}(C)^{\!\le\delta}|
&\le \sum_{\ell=k_1}^{(n\!-\!1)/2}\!\sum_{d\in{\cal D}_\ell}
|{\cal K}(C)_d|
\le \sum_{\ell=k_1}^{(n\!-\!1)/2}\!\sum_{d\in{\cal D}_\ell} |K^*|\big/q^{\ell}
=\sum_{\ell=k_1}^{(n\!-\!1)/2} |{\cal D}_\ell|\!\cdot\!|K^*|\big/q^{\ell} \\
&\le \sum_{\ell=k_1}^{(n\!-\!1)/2}\!\! |K^{\!*}|\!\cdot\!
q^{4\ell h_q(\delta)+\frac{\ell\log_q n}{k_1}}/q^\ell
=\sum_{\ell=k_1}^{(n\!-\!1)/2}\!\! |K^{\!*}|\!\cdot\!
q^{-4\ell\big(\frac{1}{4}- h_q(\delta)-\frac{\log_q n}{4k_1}\big)}.
\end{align*}
Because $\frac{1}{4}- h_q(\delta)-\frac{\log_q n}{4k_1}>0$ and
$\ell\ge k_1$, we further get
\begin{align*}
\big|{\cal K}(C)^{\!\le\delta}\big|
\le \sum_{\ell=k_1}^{(n-1)/2}
q^{-4k_1\big(\frac{1}{4}- h_q(\delta)-\frac{\log_q n}{4k_1}\big)}
|K^{\!*}|
\le q^{-4k_1\big(\frac{1}{4}- h_q(\delta)\big)+2\log_q n}
|K^{\!*}|.
\end{align*}
The last inequality holds since $\frac{n-1}{2}-k_{1}+1\leq n=q^{\log_{q}n}$.
Further, $\frac{1}{4}-h_q(\delta)>0$ and $2k_1\ge\lambda(n)$. So
\begin{align*}
\big|{\cal K}(C)^{\!\le\delta}\big|
\le |K^{\!*}|\!\cdot\!
 q^{-4k_1\big(\frac{1}{4}- h_q(\delta)\big)+2\log_q n}
\le |K^{\!*}|\!\cdot\!
 q^{-2\lambda(n)\big(\frac{1}{4}- h_q(\delta)\big)+2\log_q n}.
\end{align*}
That is,
$
\big|{\cal K}(C)^{\le\delta}\big| \le |K^{\!*}|\!\cdot\!
q^{-2\lambda(n)\big(\frac{1}{4}- h_q(\delta)
	-\frac{\log_q n}{\lambda(n)}\big)}.
$
\end{proof}

\begin{remark} \label{rk n_1,...} \rm
Let ${\cal P}$ be the set of all primes, and
${\cal P}_t$ the set of primes less than or equal to $t$.
We denote ${\cal G}_t
=\{p\in{\cal P}\,|\, q< p\le t,\,
 {\rm ord}_{{\Bbb Z}_p^\times}(q)\ge(\log_q t)^2\}$,
and denote ${\cal G}=\bigcup_{t=1}^{\infty}{\cal G}_t$.
Then the density of ${\cal G}$ is
$\lim\limits_{t\to\infty}|{\cal G}_t|\big/|{\cal P}_t|=1$,
see \cite[Lemma~II.6]{FL20}. Hence, by Lemma~\ref{lem lambda=min}, there are
positive  odd integers $n_1,n_2,\cdots $ with every
$n_i$ coprime to $q$ and $n_i\to\infty$ such that
\begin{align}\label{eq n_1,...}
  \lim_{i\to\infty}\frac{\log_q n_i}{\lambda(n_i)}=0.
\end{align}
\end{remark}

\begin{theorem} \label{thm asymptotically good}
Let $\delta$ be as in Eq.\eqref{eq delta in...}, and $n_1,n_2,\cdots$
as in Eq.\eqref{eq n_1,...}. Then
there are consta-dihedral code
$C^{(i)}$ of length $2n_i$, for $i=1,2,\cdots$, such that

{\bf(1)} the length $2n_i$ of $C^{(i)}$ is going to infinity;

{\bf(2)}
${\rm R}(C^{(i)})=\frac{1}{2}-\frac{1}{2n_i}$ for $i=1,2,\cdots$;

{\bf(3)}
the relative minimum distance $\Delta(C^{(i)})>\delta$ for $i=1,2,\cdots$;

\noindent
hence the code sequence $C^{(1)}, C^{(2)},\cdots$ is asymptotically good.
\end{theorem}

\begin{proof}
Since $\frac{1}{4}-h_q(\delta)>0$,
by dropping finitely many terms (if necessary),
we can further assume that the sequence
$n_1,n_2,\cdots$ satisfy Eq.\eqref{eq n_1,...} and that
$\frac{1}{4}-h_q(\delta)-\frac{\log_q n_i}{\lambda(n_i)}>0$,
for \,$i=1,2,\cdots$.
In Lemma~\ref{lem K^<delta}, take $n=n_i$, we get
$$
\lim_{i\to\infty} \frac{\big|{\cal K}(C)^{\le\delta}\big|}{|K^*|}
 \le \lim_{i\to\infty}q^{-2\lambda(n_{i})\big(\frac{1}{4}
 	-h_q(\delta)-\frac{\log_q n_i}{\lambda(n_i)}\big)}=0.
$$
Thus we can take $\beta^{(i)}\in K^*\backslash \mathcal K(C)^{\le\delta}$ for $i=1,2,\cdots$.
Set $C^{(i)}=C\beta^{(i)}$.
Then $C^{(i)}$ is a consta-dihedral code of length $2n_i$
and the obtained code sequence
\begin{align}\label{eq C^i}
	C^{(1)}, \, C^{(2)},\,\cdots
\end{align}
satisfy the statements (1), (2) and (3).
\end{proof}

The density of the set ${\cal G}$
of primes in Remark~\ref{rk n_1,...} equals $1$,
so that we can take some subsets of ${\cal G}$ which
satisfy more conditions.
\begin{lemma} [{\cite[Corollary II.8]{FL20}}] \label{le more on n_1,...}
There are positive odd integers $n_1,n_2, \cdots$
with every $n_i$ coprime to $q$ and $n_i\to\infty$ such that
\begin{align} \label{eq q odd order}
 \lim_{i\to\infty}\frac{\log_q n_i}{\lambda(n_i)}=0;\qquad
 \mbox{${\rm ord}_{{\Bbb Z}_{n_i}^\times}(q)$ is odd},
 ~~\forall\,i=1,2,\cdots. 	
\end{align}
\end{lemma}

\begin{theorem}
The self-orthogonal consta-dihedral codes over any finite field $F$
are asymptotically good.
\end{theorem}

\begin{proof}
Taking $n_1,n_2,\cdots$ as in Eq.\eqref{eq q odd order}.
By Lemma~\ref{a^T=bar a}(2), for the cyclic group~$H$ of order $n_i$,
any primitive idempotent $e$ of $FH$ other than $e_0$
satisfies that $\overline e\ne e$.
By Lemma~\ref {lem C_t 1=e+bar e} and Corollary~\ref{cor C=C cap A+...}(4),
the consta-dihedral code~$C^{(i)}$ with length $2n_i$ in Eq.\eqref{eq C^i}
is self-orthogonal.
\end{proof}

\begin{lemma} \label{rk more on n_1,...}
 Assume that $q$ is odd and $4\nmid(q-1)$.
Then there are positive odd integers
$n_1,n_2,\cdots$ with every $n_i$ coprime to $q$
and $n_i\to\infty$ such that
\begin{align} \label{eq T=...}
	\lim_{i\to\infty}\frac{\log_q n_i}{\lambda(n_i)}=0;\quad
	-1\in\langle q\rangle_{{\Bbb Z}_{n_i}^\times}
	\mbox{ and }
	2\,\Vert\,{\rm ord}_{{\Bbb Z}_{n_i}^\times}(q),
	~for~  \,i=1,2,\cdots, 	
\end{align}
where ``\,$2\,\Vert\,t$'' means that $2\,|\,t$ but $2^2\nmid t$.
\end{lemma}

\begin{proof}
Let ${\mathcal O}=\{\,p\in{\cal P}\,|\,
\mbox{${\rm ord}_{{\Bbb Z}_p^\times}(q)$ is odd}\,\}$
and $\overline{\mathcal O}={\cal P}\backslash{\mathcal O}$.
By the assumption,
we can write $q=r^s$ for an odd prime $r$ and an odd positive integer $s$.
By \cite[Theorem 1]{O}, the density of ${\mathcal O}$ in ${\cal P}$ equals
$\frac{1}{3}$, hence the density of $\overline{\mathcal O}$ equals
$\frac{2}{3}$. On the other hand, we consider
$${\cal T}=
\big\{\,p\in{\cal P}\,\big|\, 2\,\Vert\,(p-1) 
\,\big\}
 =\big\{\,p\in{\cal P}\,\big|\, p\equiv 3\!\!\!\pmod 4\,\big\}.
$$
By a Dirichlet's theorem on density
(cf. \cite[Ch.6 \S4 Theorem 2]{Serre}),
the density of ${\cal T}$ in ${\cal P}$ equals $\frac{1}{2}$.
Thus the density of $\overline{\mathcal O}\cap{\cal T}$
is at least $\frac{2}{3}+\frac{1}{2}-1=\frac{1}{6}$.
Hence the density of $\overline{\mathcal O}\cap{\cal T}\cap{\cal G}$
is at least $\frac{1}{6}$, where ${\cal G}$ is defined in Remark~\ref{rk n_1,...}.
For any $p\in\overline{\mathcal O_t}\cap{\cal T}\cap{\cal G}$,
we have that $-1\in\langle q\rangle_{{\Bbb Z}_{p}^\times}$
\big(because: ${\rm ord}_{{\Bbb Z}_{p}^\times}(q)$ is even and
$-1$ is the unique element of order $2$ in ${\Bbb Z}_p^\times$\big),
and $2\,\Vert\,{\rm ord}_{{\Bbb Z}_{p}^\times}(q)$
\big(because: ${\rm ord}_{{\Bbb Z}_{p}^\times}(q)\,|\,(p-1)$
 but $4\nmid(p-1)$\big).
Thus, there are positive odd integers
$n_1,n_2,\cdots$ with every $n_i$ coprime to $q$
and $n_i\to\infty$ such that
Eq.\eqref{eq T=...} holds.
\end{proof}

\begin{theorem} \label{thm LCD good}
Assume that $q$ is odd and $4\nmid(q-1)$ (i.e. $q\equiv 3~({\rm mod}~4)$).
Then the LCD consta-dihedral codes over $F$ are asymptotically good.
In particular, LCD quasi-cyclic codes of index $2$ over $F$
 are asymptotically good.
\end{theorem}

\begin{proof}
Take $n_1,n_2,\cdots$ as in Eq.\eqref{eq T=...}.
By Theorem~\ref{thm LCD case},
the  $C^{(i)}$ with length~$2n_i$ in Eq.\eqref{eq C^i}
is an LCD consta-dihedral code.
Thus the LCD consta-dihedral code sequence Eq.\eqref{eq C^i} is asymptotically good.
By Eq.\eqref{eq FH oplus ...}, any consta-cyclic code 
is a quasi-cyclic code of index $2$. So the ``In particular'' part holds.
\end{proof}

\subsection{Self-dual consta-dihedral codes}

In this subsection we always assume that
$q$ is even or $4\,|\,(q-1)$, i.e., $q\,{\not\equiv}\,3~({\rm mod}~4)$.
Keep the notation in Theorem~\ref{thm self-dual case}:
\begin{itemize}
\item
  $\widehat C=C_0\oplus C=C_0\oplus C_1\oplus\cdots\oplus C_m$
where $C_0=A_{0}(re_{0}+e_{0}\dot v)$ is defined in Lemma~\ref{F-e0}(2);
\item
For any $\beta=e_0+\beta_1+\cdots+\beta_m\in K^*$,
the consta-dihedral code
$$
 \widehat C\beta=C_0\oplus C_1\beta_1\oplus\cdots\oplus C_m\beta_m
$$
is self-dual; in particular, the rate ${\rm R}(\widehat C\beta)=\frac{1}{2}$.
\end{itemize}
We will find the $\beta$ such that
the relative minimal distance $\Delta(\widehat C\beta)>\delta$.

 Note that $\cda=A_0\oplus A$.
For any $\widehat d=d_0+d\in \cda$
with $d_0\in A_0$ and $d\in A$,
if $d_0\notin C_0$, then $\widehat d\notin \widehat C\beta$
for any $\beta\in K^*$.

\begin{lemma}  \label{lem K_hat d}
Assume that $0\ne\widehat d=d_0+d\in C_0\oplus A$
with $d_0\in C_0$ and $d\in A$.
Set $\mathcal K (\widehat C)_{\widehat d}
 =\{\beta\in K^*\;|\;\widehat d\in\widehat C\beta \}$.
Then
$$
\big|\mathcal K(\widehat C)_{\widehat d}\big|\le |K^*|\big/q^{\ell_d}.
$$
where $\ell_d$ is defined in Eq.\eqref{eq ell_d<=}.
\end{lemma}

\begin{proof}
It is clear that $\widehat d\in\widehat C\beta$ if and only if
$d\in C\beta$. So this lemma follows from Lemma~\ref{lem K_d}
immediately.
\end{proof}

\begin{lemma} \label{lem hat K^<delta}
Let ${\mathcal K(\widehat C)}^{\!\le\delta}=
\{\,\beta\in K^*\,|\,\Delta(\widehat C\beta)\le\delta\,\}$.
If $\frac{1}{4}-h_q(\delta)-\frac{\log_q n}{\lambda(n)}>0$,
then
\begin{align} \label{eq hat K^<=...}
		\big|\mathcal K(\widehat C)^{\le\delta}\big| \le |K^{*}|\!\cdot\!
		q^{-2\lambda(n)\big(\frac{1}{4}- h_q(\delta)
			-\frac{\log_q n}{\lambda(n)}\big)+h_q(\delta)}.
\end{align}
\end{lemma}

\begin{proof}
For $k_1\le\ell\le\frac{n-1}{2}$,
we extend the notation ${\cal D}_\ell$ in Eq.\eqref{eq A_ell=} and set
\begin{align*}
\widehat{\cal D}_\ell
=\big\{\,\widehat d=d_0+d\;\big|\; d_0\in C_0,\,d\in A, \,
	0<{\rm w}(\widehat d\,)/2n\le\delta,\,\ell_d=\ell\,\big\}.
\end{align*}
With ${\cal A}_\ell$ defined in Eq.\eqref{eq A_ell=}, we have that
\begin{align} \label{eq hat D_ell}
\widehat{\cal D}_\ell\subseteq \bigcup_{A_\omega\in{\cal A}_\ell}
		(C_0\oplus A_\omega)^{\!\le\delta}
\quad\mbox{and}\quad
{\mathcal K(\widehat C)}^{\!\le\delta}=
\bigcup_{\ell=k_1}^{(n-1)/2}
\bigcup_{\widehat d\in\widehat{\cal D}_\ell}
  {\mathcal K(\widehat C)}_{\widehat d}.
\end{align}
For $A_\omega\in{\cal A}_\ell$,
we have that
$|(C_0\oplus{\cal A}_\omega)^{\!\le\delta}|\le q^{(4\ell+1) h_q(\delta)}$
since $\dim_{F}(C_0\oplus A_\omega)=4\ell+1$;
see Lemma~\ref{q01}.
By Eq.\eqref{eq hat D_ell} and Eq.\eqref{eq |cal A|<=}, we have
\begin{align*}
|\widehat{\cal D}_\ell|&\le \sum_{A_\omega\in{\cal A}_\ell}
|(C_0\oplus{\cal A}_\omega)^{\!\le\delta}|
\le n^{\frac{\ell}{k_1}} q^{(4\ell+1) h_q(\delta)}
		=q^{4\ell h_q(\delta)+\frac{\ell\log_q n}{k_1}+h_q(\delta)}.
\end{align*}
For $\widehat d=d_0+d\in\widehat D_\ell$, we have $\ell_d=\ell$.
By Eq.\eqref{eq hat D_ell} and Lemma~\ref{lem K_hat d},
\begin{align*}
&\big|{\mathcal K(\widehat C)}^{\!\le\delta}\big|
\le \sum_{\ell=k_1}^{(n\!-\!1)/2}\!\sum_{\widehat d\in\widehat{\cal D}_\ell}
\big|{\mathcal K(\widehat C)}_{\widehat d}\big|
\le\sum_{\ell=k_1}^{(n\!-\!1)/2}\!\sum_{\widehat d\in\widehat{\cal D}_\ell}
|K^*|\big/q^{\ell}\\
&=\sum_{\ell=k_1}^{(n\!-\!1)/2}\! |\widehat D_\ell|\!\cdot\!|K^*|\big/q^{\ell}
\le \sum_{\ell=k_1}^{(n-1)/2}\!\! |K^{\!*}|\!\cdot\!
q^{4\ell h_q(\delta)+\frac{\ell\log_q n}{k_1}+h_q(\delta)}/q^\ell.
\end{align*}
Therefore,
\begin{align*}
\big|{\mathcal K(\widehat C)}^{\!\le\delta}\big|
&=\sum_{\ell=k_1}^{(n-1)/2}\!\! |K^{\!*}|\!\cdot\!
q^{-4\ell\big(\frac{1}{4}- h_q(\delta)-\frac{\log_q n}{4k_1}\big)
+h_q(\delta)}\\
&\le \sum_{\ell=k_1}^{(n-1)/2} \! |K^{\!*}|\!\cdot\!
q^{-4k_1\big(\frac{1}{4}- h_q(\delta)-\frac{\log_q n}{4k_1}\big)
+h_q(\delta)}.
\end{align*}
By the same argument for Lemma~\ref{lem K^<delta},
we can get
$$\sum_{\ell=k_1}^{(n-1)/2}
q^{-4k_1\big(\frac{1}{4}- h_q(\delta)-\frac{\log_q n}{4k_1}\big)
+h_q(\delta)}
\leq q^{-2\lambda(n)\big(\frac{1}{4}- h_q(\delta)-\frac{\log_q n}{\lambda(n)}\big)
+h_q(\delta)}.$$
We are done.
\end{proof}

\begin{theorem} \label{thm self-dual good1}
Assume that $q$ is even or $4\,|\,(q-1)$.
Let $\delta$ be as in Eq.\eqref{eq delta in...}, and $n_1,n_2,\cdots$
as in Eq.\eqref{eq n_1,...}. Then
there are self-dual consta-dihedral codes
$\widehat C^{(i)}$ of length $2n_i$ 
such that $\Delta(\widehat C^{(i)})>\delta$ for all $i=1,2,\cdots$;
hence the code sequence $\widehat C^{(1)}, \widehat C^{(2)},\cdots$
is asymptotically good.
\end{theorem}

\begin{proof}
In Lemma~\ref{lem hat K^<delta}, we set $n=n_i$, so
$$
\lim_{i\to\infty}\frac{\big|\mathcal K(\widehat C)^{\le\delta}\big|}{|K^*|}
\le \lim_{i\to\infty}q^{-2\lambda(n_i)\big(\frac{1}{4}
	-h_q(\delta)-\frac{\log_q n_i}{\lambda(n_i)}\big)+h_q(\delta)}
= 0.
$$
We can take
$\beta^{(i)}\in K^*\backslash \mathcal K(\widehat C)^{\le\delta}$ for $i=1,2,\cdots$.
Set $\widehat C^{(i)}=\widehat C\beta^{(i)}$.
Then $\widehat C^{(i)}$ is a self-dual consta-dihedral code
of length $2n_i$ and the code sequence
\begin{align}\label{eq hat C^i}
 \widehat C^{(1)}, \, \widehat C^{(2)},\,\cdots
\end{align}
satisfies that $\Delta(\widehat C^{(i)})>\delta$ for $i=1,2,\cdots$.
\end{proof}

We get the following immediately (cf. the proof of Theorem~\ref{thm LCD good}).

\begin{theorem} \label{thm self-dual good}
Assume that $q$ is even or $4\,|\,(q-1)$ (i.e. $q\,{\not\equiv}\,3~({\rm mod}~4)$).
Then the self-dual consta-dihedral codes over $F$ are asymptotically good.
In particular, self-dual quasi-cyclic codes of index $2$ over $F$
 are asymptotically good.
\end{theorem}

It is known that if $q$ is odd then LCD dihedral codes of rate $\frac{1}{2}$
are asymptotically good, see \cite[Theorem 1.2]{FL20}.
We get the following consequence.

\begin{corollary} \label{cor q equiv 1 mod 4}
If $q\equiv 1~({\rm mod}~4)$, then the
self-dual quasi-cyclic codes of index~$2$ over $F$ and
the LCD quasi-cyclic codes of index $2$ over $F$ of rate $\frac{1}{2}$
are both asymptotically good.
\end{corollary}

\section{Remarks on dihedral codes}\label{Remarks on}

The purpose of this section is to correct some mistakes in~\cite{FL20}.
We begin with a subtle remark.
\begin{remark} \rm
Lemma~\ref{dual-C}(5) (i.e., \cite[Lemma II.4(5)]{FL20}) provides
an efficient technique to evaluate the orthogonality of group codes, i.e.,
for $FG$-codes $C, D$,
\begin{align} \label{eq C bar D}
\langle C,D\rangle=0 ~~~ \iff ~~~ C\overline D=0.
\end{align}
However, a subtle point is that the following is {\em incorrect}:
\begin{align} \label{eq bar D C}
\langle C,D\rangle=0 ~~~ \implies ~~~ \overline D C=0.
\end{align}
Here is a counterexample for Eq.\eqref{eq bar D C}.

\medskip
{\bf Example}.
Take $|F|=7$, $n=3$,
$G=\langle u,v\,|\,u^3=1=v^2, vuv^{-1}=u^{-1}\rangle
=H\rtimes\langle v\rangle$ be the dihedral group of order $6$,
where $H=\{1,u,u^2\}$ is the cyclic group of order $3$.
Then $\frac{1}{3}$ equals 5 in $F$
and 2, 4 are primitive $3$'th roots of unity.
 $e_0, e, \bar e$ are all primitive idempotents of $FH$, where
$$
e_0=5(1+u+u^2), \quad e=5(1+2u+4u^2),\qquad  \bar e=5(1+4u+2u^2).
$$
Denote $FG=A_0\oplus A_1$, where
$$ A_0=FG e_0,
\quad
 A_1=FG(e+\bar e)=FHe\oplus FH\bar e\oplus FHev\oplus FH\bar e v,
$$
are minimal ideals of $FG$.
Take $C=D=A_1e=FHe\oplus FH\bar ev$. Then
$$C\overline D=A_1e\cdot\overline{A_1e}
=A_1e\bar e\overline A_1=A_10\,\overline A_1=0.
$$
By Eq.\eqref{eq C bar D}, $\langle C,D\rangle=0$.
However, $\overline D C=\bar e\overline A_1\cdot A_1e=\bar eA_1e \neq 0$,
because we can choose an element
$\bar e(\bar e v)e\in\bar eA_1e$
and
$\bar e(\bar e v)e=\bar e v e =\bar e \bar e v=\bar e v\ne 0$.
\end{remark}

\begin{remark} \label{rk issue} \rm
Turn to the mistakes of \cite{FL20}.
The main issue in~\cite{FL20} is that
\begin{itemize}
\item[(I)]
In the proofs of \cite[Theorem~IV.3]{FL20} and \cite[Theorem~IV.5]{FL20},
some citations of \cite[Lemma II.4(5)]{FL20} (i.e., Eq.\eqref{eq C bar D}) are in fact
misuses of the incorrect version Eq.\eqref{eq bar D C}.
\end{itemize}
We first show the effects of the issue, then explain how to address it.

\cite[Theorem IV.3]{FL20} considers the case that ${\rm char}(F)=2$.
Though the incorrect Eq.\eqref{eq bar D C} was misused in its proof,
\cite[Theorem IV.3]{FL20} is itself correct.
Because: if ${\rm char}(F)=2$,
then $-1=1$ and the consta-dihedral group algebra
is identified with the dihedral group algebra,
i.e., $\cda=FG$, cf. Eq.\eqref{eq def F*G};
hence all the results in this paper
are applied to $FG$ and to dihedral codes provided ${\rm char}(F)$ is even.
Thus, \cite[Theorem~IV.3]{FL20} is a consequence of
Theorem~\ref{thm self-dual case} (by taking even $q$) of this paper.

\cite[Theorem IV.5]{FL20} considers the case that ${\rm char}(F)$ is odd,
and consists of two parts:
(1) the case that ${\rm ord}_{{\Bbb Z}_n^\times}(q)$ is odd;
(2) the case that $-1\in\langle q\rangle_{{\Bbb Z}_n^\times}$.
For \cite[Theorem~IV.5(2)]{FL20},
though there were gaps in the proof, its conclusion is still correct
and proved in \cite[Lemma 8.6(2)]{FL22b}.
For \cite[Theorem~IV.5(1)]{FL20}
(the case that ${\rm char}(F)$ is odd
and ${\rm ord}_{{\Bbb Z}_n^\times}(q)$ is odd), however,
it is unlucky that the misuse of the incorrect Eq.\eqref{eq bar D C}
results in an incorrect conclusion.
\end{remark}

\begin{remark}\label{rk by matrix} \rm
The issue described in Remark~\ref{rk issue} implies that
on some occasions 
Eq.\eqref{eq C bar D} is not enough to recognize
orthogonality of group codes. In this paper we developed a technique
to recognize the orthogonality of group codes by matrix computations,
e.g., see the proofs of Lemma~\ref{lem C_t 1=e+bar e} and
Lemma~\ref{lem C_t 1=e}.
That is one of the contributions of this paper.
By this technique,
to look for a correct version of \cite[Theorem IV.5(1)]{FL20},
we begin with the dihedral group algebra version of
Lemma~\ref{A-M bar e not e}.
\end{remark}

Keep the assumption Eq.\eqref{eq G dihedral} and Eq.\eqref{eq FH=...}.

\begin{lemma}\label{dihedral bar e not e}
Let e be a primitive idempotent of $FH$ with $\overline e\ne e$.
Then $\tilde F:=FHe$ is a field extension over F,  $e+\overline e$
is a primitive central idempotent of $FG$ and:

{\bf(1)}
The ideal $FG(e+\overline{e})
=F He\oplus FH \bar e\oplus FH e v\oplus FH\bar e v
\cong{\rm M}_2(\tilde F)$.

{\bf(2)} With the isomorphism in (1),
if $f\in FG(e+\overline{e})$ corresponds to the matrix
$\begin{pmatrix}a_{11}& a_{12}\\ a_{21} & a_{22}\end{pmatrix}
\in{\rm M}_2(\tilde F)$, then
$\overline f$ corresponds to the matrix
$\begin{pmatrix}a_{22}& a_{12}\\ a_{21} & a_{11}\end{pmatrix}$.
\end{lemma}

\begin{proof}
Similarly to the proof of Lemma~\ref{A-M bar e not e},
$\tilde F:=FHe$ is a field extension over F
and  $e+\overline e$ is a primitive central idempotent of $FG$.
So $FG (e+\bar e)=F He\oplus FH e v \oplus FH \bar e\oplus FH\bar e v$ is an ideal of $FG$.
Define a map:
\begin{align} \label{eq Corre 1 dihedral}
\begin{array}{ccc}
{\rm M}_2(\tilde F) & \mathop{\longrightarrow}\limits^{\cong} &
F He\oplus FH e v \oplus FH \bar e\oplus FH\bar e v,
\\[3pt]
\begin{pmatrix}a_{11} & a_{12}\\ a_{21} & a_{22} \end{pmatrix}
& \longmapsto &
a_{11}e + a_{12}\,e\, v + \overline{a_{21}}\,\bar e\, v + \overline{a_{22}}\,\bar e,
\end{array}
\end{align}
which is obviously a linear isomorphism.
 For $a_{ij},b_{ij}\in \widetilde F$, $1\le i,j\le 2$,
noting that $ v a_{ij}=\overline{a_{ij}}\, v$ and $ v^2=1$, we have
\begin{eqnarray*}
\begin{split}
\big(a_{11}e + a_{12}\,e\, v + \overline{a_{21}}\,\bar e v
 + \overline{a_{22}}~\overline{e}\big)
\big(b_{11}e + b_{12}e\, v + \overline{b_{21}}\,\bar e v
 +\overline{b_{22}}~\overline{e})\\
= (a_{11}b_{11}+a_{12}b_{21})e + (a_{11}b_{12}+a_{12}b_{22})e\, v \,+ \\
 \overline{(a_{21}b_{11}+a_{22}b_{21})}~\overline e \, v
 +\overline{(a_{21}b_{12}+a_{22}b_{22})}~\overline e .
 \end{split}
\end{eqnarray*}
So, Eq.\eqref{eq Corre 1 dihedral} is an $\tilde F$-algebra isomorphism.

For ${a_{11}e + a_{12}\,e\, v + \overline{a_{21}}\,\bar e v
 + \overline{a_{22}}~\overline{e}}\in{F He\oplus FH e v \oplus FH \bar e\oplus FH\bar e v}$, \\
 we have (note that $\overline{v}= v$):
\begin{align*}
\overline{a_{11}e + a_{12}e\, v +
  \overline{a_{21}}\,\bar e\, v + \overline{a_{22}}\,\bar e}
&=\overline{a_{11}}\,\bar e + \overline{v}\,\overline{a_{12}}\,\bar e
  + \overline{v}\,a_{21}\, e + a_{22}\, e.\\
&= a_{22}\, e +\,a_{12}\, e\, v
  + \overline{a_{21}}\,\bar e\, v + \overline{a_{11}}\,\bar e.
\end{align*}
Thus, the bar image of ${a_{11}e + a_{12}\,e\, v + \overline{a_{21}}\,\bar e v
 + \overline{a_{22}}~\overline{e}}$ corresponds the matrix:
\begin{align} \label{bar to matrix dihedral}
\overline{a_{11}e + a_{12}e\,v +
  \overline{a_{21}}\,\bar e\,v + \overline{a_{22}}\,\bar e}
 ~~\longleftrightarrow~~
 \begin{pmatrix}a_{22} & a_{12}\\ a_{21} & a_{11} \end{pmatrix}.
\end{align}
We are done.
\end{proof}

Recall that (Lemma~\ref{lem matrix size 2}(2)),
there are altogether $|\tilde F|+1$ simple left ideals of  ${\rm M}_2(\tilde F)$
with generators:
\begin{align} \label{left ideal of A}
 \begin{pmatrix} a&1\\0&0 \end{pmatrix}, \quad a\in\tilde F;
 \qquad\mbox{or}\quad \begin{pmatrix} 1&0\\0&0 \end{pmatrix}.
\end{align}

\begin{lemma} \label{lem C_t 1=e+bar e dihedral}
Let notation be as above in Lemma~\ref{dihedral bar e not e}.
Denote $A=FG(e+\bar e)$ for short.
Let $f_{ab}\in A$  be the element corresponding to
$\begin{pmatrix} a&b\\0&0 \end{pmatrix}\in{\rm M}_2(\tilde F)$.
Then

{\bf(1)} If ${\rm char}(F)=2$, then $\langle Af_{ab},Af_{ab}\rangle=0$,
for $\,a,b\in\tilde F$.

{\bf(2)} If ${\rm char}(F)$ is odd, then
$\langle Af_{ab},Af_{ab}\rangle=0$ if and only if $ab=0$;
in particular, there are exactly two simple left ideals of $A$ which is self-dual in $A$,
and the other $|\tilde F|-1$ simple left ideals of $A$ are LCD in $A$.
\end{lemma}

\begin{proof}
By Lemma~\ref{dual-C}(5), $\langle Af_{ab},Af_{ab}\rangle=0$ if and only if
$f_{ab}\overline{f_{ab}}=0$.
By Lemma~\ref{dihedral bar e not e}
(Eq.\eqref{eq Corre 1 dihedral} and Eq.\eqref{bar to matrix dihedral}),
$f_{ab}\overline{f_{ab}}=0$ if and only if
$$
\begin{pmatrix} a&b\\0&0 \end{pmatrix}\begin{pmatrix} 0&b\\0&a \end{pmatrix}
=\begin{pmatrix} 0&2ab\\0&0 \end{pmatrix}=0.
$$
If ${\rm char}(F)=2$, it is always true that $2ab=0$, hence (1) holds.
Next assume that ${\rm char}(F)\ne 2$.
Then $2ab=0$ if and only if $ab=0$.
In Eq.\eqref{left ideal of A}, there exactly two cases such that $ab=0$, i.e.,
$\begin{pmatrix} 0&1\\0&0 \end{pmatrix}$ and
$\begin{pmatrix} 1&0\\0&0 \end{pmatrix}$.
Hence (2) follows.
\end{proof}

\begin{remark} \label{q odd dihedral} \rm
Assume that $q$ is odd and ${\rm ord}_{{\Bbb Z}_n^\times}(q)$ is odd.
 For all primitive idempotents
$e_0=\frac{1}{n}\sum_{i=0}^{n-1}u^i,\;e_1,\cdots,e_\ell$
of $FH$, except for $e_0$, the other primitive idempotents
are pairwise partitioned (see Lemma~\ref{a^T=bar a}(2)):
$$
 e_1,~\overline{e_1}, ~\cdots, ~e_m,~\overline{e_m}.
$$

(1) Set $A_0=FGe_0$ and $A_t=FG(e_t+\overline{e_t})$ for $t=1,\cdots,m$.
By Lemma~\ref{dihedral bar e not e},
$$
FG=A_0\oplus A_1 \oplus\cdots\oplus A_m;\qquad
A_t\cong{\rm M}_2(F_t), ~~ t=1,\cdots,m.\;
$$
where $F_t=FH e_t$ is a field extension over $F$, denote $k_t=|F_t:F|$.
By Eq.\eqref{eq C bar D}, $\langle A_i,A_j\rangle=0$ for $0\le i\ne j\le m$;
hence Corollary~\ref{cor C=C cap A+...} is still valid for
$FG=A_0\oplus A_1 \oplus\cdots\oplus A_m$.

(2)
Let $\widehat e_0=e_0+e_0v\in A_0$, then
$\langle A_0\widehat e_0, A_0\widehat e_0\rangle
= A_0\widehat e_0\overline {\widehat e_0} \overline{\widehat A_0}\neq 0$
since $\widehat e_0\overline{\widehat e_0}=2\widehat e_0\ne 0$,
hence $A_0\widehat e_0$ is LCD in $A_0$.
\end{remark}

Thus the following is the correct version of \cite[Theorem IV.5(1)]{FL20}.

\begin{theorem} \label{thm correct}
Let notation be as  in Remark~\ref{q odd dihedral}.
We consider the following dihedral codes:
\begin{align} \label{eq C_ab}
 C_{ab}\!=\!A_0\widehat e_0\oplus A_1f_{a_1b_1}\oplus\cdots\oplus A_mf_{a_mb_m},
 ~~ 0\!\ne\! (a_t,b_t)\in F_t^2, ~t\!=\!1,\cdots\!,m.
\end{align}
Then $\dim_F(C_{ab})=n$, and

{\bf(1)} The number of the dihedral codes in Eq.\eqref{eq C_ab} equals
  $\prod_{t=1}^{m}(q^{k_t}+1)$.

{\bf(2)} The number of the dihedral codes in Eq.\eqref{eq C_ab} which are LCD
 equals  $\prod_{t=1}^{m}(q^{k_t}-1)$.
\end{theorem}

\begin{proof}
By \cite[Lemma~III.2(3)]{FL20}, $A_0\widehat e_0$
is an $1$-dimensional ideal of~$A_0$, hence
$\dim_F(C_{ab})=1+\dim_{F} A_1f_{a_1b_1}+\cdots +\dim_{F}A_mf_{a_mb_m}=n$,
where $\dim_{F}A_tf_{a_tb_t}=2k_{t}$, $t=1,\cdots,m$.
By Eq.\eqref{left ideal of A} and Lemma~\ref{lem C_t 1=e+bar e dihedral}(2),
 we get (1) and (2) at once.
\end{proof}

\section{Conclusion}\label{Conclusion}

We studied the consta-dihedral codes, and addressed an issue of the
reference~\cite{FL20} which is a research on dihedral codes.

To investigate the algebraic property of the consta-dihedral codes,
the existing methods, cf. in \cite{FL20},
are not enough to recognize the orthogonality of group codes;
so we developed a technique to evaluate the orthogonality of group codes
by matrix computation.
We characterized the algebraic structure of consta-dihedral group algebras.
By the algebraic structure and with the technique mentioned just now,
we constructed a class of consta-dihedral codes which possess good algebraic property
(self-orthogonal, or LCD).

Next, we showed the existence of asymptotic good sequences
of the consta-dihedral codes in the class we constructed.
Instead of probabilistic methods, in the class
we counted directly the number of the consta-dihedral codes with
bad asymptotic property. This number is much less than the
total quantity of the class. In this way we obtained
(recall that $F$ is a finite field with $|F|=q$):

--- If $q$ is even or $4\,|\,(q-1)$, then the self-dual consta-dihedral codes
over $F$ are asymptotically good.

--- If $q$ is odd and $4\nmid(q-1)$, then the LCD consta-dihedral codes
over $F$ are asymptotically good.

Finally, with the help of the technique developed in this paper,
we addressed the issue in \cite{FL20}
and obtained the correct version of the false theorem
 \cite[Theorem~IV.5(1)]{FL20}.

\section*{Acknowledgements}

\end{document}